\documentclass[12pt,a4paper]{article} 
\usepackage[left=38mm,right=28mm, top=35mm, bottom=35mm]{geometry}

\usepackage[english]{babel}
\usepackage[babel,english=american]{csquotes}
\usepackage{color, comment}
\usepackage{amssymb,amsfonts,amsmath,amsthm,bbm}
\usepackage{latexsym}
\usepackage{graphicx}
\usepackage{multirow}
\usepackage{stmaryrd}
\usepackage{dsfont}
\usepackage{nicefrac}
\usepackage{enumerate}
\usepackage{mathrsfs}
\usepackage{natbib}
\usepackage[reftex]{theoremref}

\newcommand{\R}{\mathbb{R}}

\newcommand{\N}{\mathbb{N}}

\newcommand{\Linf}[2]{L_{#2}^\infty(#1)}

\newcommand{\CB}{\mathcal{B}}

\newcommand{\CE}{\mathcal{E}}
\newcommand{\CF}{\mathcal{F}}
\newcommand{\CG}{\mathcal{G}}
\newcommand{\CH}{\mathcal{H}}
\newcommand{\CI}{\mathcal{I}}
\newcommand{\CJ}{\mathcal{J}}

\newcommand{\CM}{\mathcal{M}}

\newcommand{\CR}{\mathcal{R}}
\newcommand{\CS}{\mathcal{S}}
\newcommand{\CT}{\mathcal{T}}

\newcommand{\CX}{\mathcal{X}}

\newcommand{\PW}{\mathds{P}}

\newcommand{\WR}{\Omega}
\newcommand{\om}{\omega}

\newcommand{\LG}{\Lambda_{\CG}}

\newcommand{\vecone}{\mathbf{1}_d}
\newcommand{\veczero}{\mathbf{0}_d}

\newcommand{\WH}{\widetilde{H}}

\newcommand{\Wrho}{\widetilde{\rho}}

\newcommand{\Wu}{\widetilde{u}}
\newcommand{\WU}{\widetilde{U}}

\newcommand{\rhoo}{\bar{\rho}}

\newcommand{\eqd}{\stackrel{\text{d}}{=}}

\newcommand{\ind}{\mathbbmss{1}}

\newcommand{\prefto}{\succcurlyeq}

\newcommand{\EW}[2]{\mathds{E}_{#1}\left[#2\right]}
\newcommand{\BEW}[3]{\mathds{E}_{#1}\left[\left.#2\,\right|\,#3\right]}
\newcommand{\eq}[1]{#1}

\newcommand{\eqX}{\eq{X}}

\newcommand{\eqXom}{\eq{X(\om)}}
\newcommand{\eqx}{\eq{x}}

\newcommand{\eqY}{\eq{Y}}
\newcommand{\eqYom}{\eq{Y(\om)}}

\DeclareMathOperator*{\esssup}{esssup}

\DeclareMathOperator{\id}{id}
\DeclareMathOperator{\Imag}{Im}

\newcommand{\supnorm}[1]{\left\|#1\right\|_{\infty}}

\numberwithin{figure}{section}
\numberwithin{table}{section}
\numberwithin{equation}{section}

\newtheorem{Theorem}{Theorem}[section]
\newtheorem{Proposition}[Theorem]{Proposition}
\newtheorem{Lemma}[Theorem]{Lemma}
\newtheorem{Corollary}[Theorem]{Corollary}
\newtheorem{Definition}[Theorem]{Definition}

\theoremstyle{definition}
\newtheorem{Example}[Theorem]{Example}

\newtheorem{Assumption}[]{Assumption}

\theoremstyle{remark}
\newtheorem{Remark}[Theorem]{Remark}

\begin{document}
\selectlanguage{english}
\title{Strongly Consistent Multivariate Conditional Risk Measures}
\author{Hannes Hoffmann\thanks{Department of Mathematics, University of Munich, Theresienstra{\ss}e 39, 80333 Munich, Germany. Emails: hannes.hoffmann@math.lmu.de, meyer-brandis@math.lmu.de and gregor.svindland@math.lmu.de.}\and
Thilo Meyer-Brandis\footnotemark[1]\and
Gregor Svindland\footnotemark[1]}

\maketitle
\begin{abstract}
We consider families of strongly consistent multivariate conditional risk measures. We show that under strong consistency these families admit a decomposition into a conditional aggregation function and a univariate conditional risk measure as introduced \cite{Hoffmann2016}. Further, in analogy to the univariate case in \cite{Follmer2014}, we prove that under law-invariance strong consistency implies that  multivariate conditional risk measures are necessarily multivariate conditional certainty equivalents. 
\\[4mm]\noindent\textbf{Keywords: }multivariate risk measures, systemic risk measures, strong consistency, systemic risk, law-invariance, conditional certainty equivalents. 
\\[4mm]\textbf{MSC 2010 classifications:} 91B30, 91G99
\end{abstract}

\section{Introduction}
Over the recent years the study of \emph{multivariate risk measures} 
\begin{equation}\label{eq:MRM}
\rho: \Linf{\CF}{d}\to\mathbb{R},
\end{equation}
that associate a risk level $\rho(X)$ to a $d$-dimensional vector $X=(X_1,...,X_d)$ of random risk factors at a given future time horizon $T$ has increasingly gained importance. Here, $\Linf{\CF}{d}$ denotes the space of $d$-dimensional bounded random vectors on a probability space $(\Omega,\CF,\PW)$, i.e.\ we restrict the analysis to bounded risk factors $X$ for technical simplicity. 


A natural extension of the static viewpoint of deterministic risk measurement in \eqref{eq:MRM} is to consider \emph{conditional risk measures} which allow for risk measurement under varying information. A conditional multivariate risk measure is a map
\begin{equation}\label{eq:condrm}\rho_\CG:\Linf{\CF}{d}\to\Linf{\CG}{},\end{equation} 
that associates to a $d$-dimensional risk factor a $\CG$-measur\-able bounded random variable, where $\CG\subseteq\CF$ is a sub-$\sigma$-algebra. We interpret $\rho_\CG(X)$ as the risk of $X$ given the information $\CG$. In the present literature, conditional risk measures have mostly been studied within the framework of univariate \emph{dynamic} risk measures, where one adjusts the risk measurement in response to the flow of information that is revealed when time elapses. For a good overview on univariate dynamic risk measures we refer the reader to \cite{Acciaio2011} or \cite{Tutsch2007}. One possible motivation to study conditional multivariate risk measures is thus the extension from univariate to multivariate dynamic risk measures, and to study the question of what happens to the risk of a system as new information arises in the course of time. In the context of multivariate risk measures, however, also a second interesting and important dimension of conditioning arises, besides dynamic conditioning: Risk measurement conditional on information in space in order to identify systemic relevant structures. In that case $\CG$ represents for example information on the state of a subsystem, and one is interested in questions of the type: How is the overall risk of the system affected, given that a subsystem is in distress? Or how is the risk of a single institution affected, given the entire system is in distress? In \cite{Follmer2014} and \cite{Follmer2014a} the authors analyze such spatial conditioning in the context of univariate conditional risk measures, so-called spatial risk measures. Another field of application where these questions are important are the systemic risk measures, which measure the risk of a financial network. In particular the systemic risk measures CoVaR of \cite{Adrian2011} or the systemic expected shortfall of \cite{Acharya2010} can be considered to be examples of conditional multivariate risk measures.

When dealing with families of conditional risk measures, a frequently imposed requirement is that the conditional risk measurement behaves consistent in a certain way with respect to the flow of information. In particular, in the literature on univariate dynamic risk measures most often the so-called \emph{strong consistency} is studied; c.f.\ \cite{Detlefsen2005,Cheridito2006,Cheridito2011,Kupper2009,Penner2007}. Two univariate conditional risk measures $\rho_\CG$ and $\rho_\CH$ with corresponding $\sigma$-algebras $\CG\subseteq\CH\subseteq\CF$ are called strongly consistent if for all $X,Y\in\Linf{\CF}{}$
\begin{equation}\label{eq:intro:strongconst}
	\rho_\CH(X)\leq \rho_\CH(Y)\Longrightarrow \rho_\CG(X)\leq \rho_\CG(Y),
\end{equation}
i.e.\ strong consistency states that if $Y$ is riskier than $X$ given the information $\CH$, then this risk preference also holds under less information. 

The purpose of this paper is to study the concept of strong consistency for multivariate conditional risk measures. Note that the motivation and interpretation of strong consistency in \eqref{eq:intro:strongconst} remains perfectly meaningful when extending to the multivariate case. In analogy to the univariate case 
we thus define strong consistency of two multivariate conditional risk measures $\rho_\CG$ and $\rho_\CH$ with $\CG\subseteq\CH\subseteq\CF$ as in \eqref{eq:intro:strongconst} for any $d$-dimensional risk vectors $X$ and $Y$ in $\Linf{\CF}{d}$.  
As a first main result we then prove that the members of any family of strongly consistent multivariate conditional risk measures are necessarily of the following from:  \begin{equation}\label{eq:aim2}
\rho_\CG(X)=\eta_\CG \left(\Lambda_\CG(X)\right),
\end{equation}
where $\eta_\CG:\Linf{\CF}{}\to\Linf{\CG}{}$ is a univariate conditional risk measure, and $\Lambda_\CG:\Linf{\CF}{d}\to\Linf{\CF}{}$ is a (conditional) aggregation function. This subclass of multivariate conditional risk measures corresponds to the idea that we first aggregate the risk factors $X$ and then evaluate the risk of the aggregated values. In fact many prominent examples of multivariate conditional risk measures are of type \eqref{eq:aim2}, for instance the Contagion Index of \cite{Cont2013} or the SystRisk of \cite{Brunnermeier2013}  from the systemic risk literature. \cite{Chen2013} were the first to axiomatically describe this intuitive type of multivariate risk measures on a finite state space, and in \cite{Kromer2016} this has been extended to general $L^p$-spaces, whereas the conditional framework was studied in \cite{Hoffmann2016}.
We also remark that in \cite{Kromer2014} the authors study consistency of risk measures over time which can be decomposed as in \eqref{eq:aim2}. However, their definition of consistency differs from ours in \eqref{eq:intro:strongconst} as they require consistency of the underlying univariate risk measure and the aggregation function in \eqref{eq:aim2} simultaneously.

A requirement on the strongly consistent family of multivariate conditional risk measures we ask for here---which is automatically satisfied in the univariate case---is that it contains a terminal risk measure $\rho_\CF:\Linf{\CF}{d}\to \Linf{\CF}{}$ under full information $\CF$. Such a terminal risk measure is nothing but a statewise aggregation rule for the components of a risk $X\in \Linf{\CF}{d}$. In the univariate case, if $X\in \Linf{\CF}{}$, there is of course no aggregation necessary. Indeed letting the terminal risk measure correspond to the identity mapping, i.e.\ $\rho_\CF=-\operatorname{id}$, we have that any univariate risk measure $\rho_\CG$ with $\CG\subseteq \CF$ is strongly consistent with $\rho_\CF$ by monotonicity, so the existence of such a terminal risk measure which is strongly consistent with the other risk measures of the family is no further restriction. In the truly multivariate case, however, it is 
very natural that also under full information there is a rule for aggregating risk over the dimensions, and the risk measures in the family should be consistent with this terminal aggregation rule. If this is the case, we show, as already mentioned, that the members of the family are necessarily of type \eqref{eq:aim2}. Indeed we show that by strong consistency the risk measures inherit a property called \emph{risk-antitonicity} in \cite{Hoffmann2016} from the terminal risk measure. This property  is the essential axiom behind allowing for a decomposition of type \eqref{eq:aim2}; see Theorem~\ref{l:riskanti2}.

Along the path to this result we characterize strong consistency in terms of a tower property. It is well-known, see e.g.\ \cite{Tutsch2007}, that for univariate conditional risk measures which are normalized on constants ($\eta_\CG(a)=-a$ for all $a\in L^\infty(\CG)$), strong consistency \eqref{eq:intro:strongconst} is equivalent to the following tower property:
\begin{equation}\label{eq:intro:tower}
\rho_\CG(X)=\rho_\CG\big(-\rho_\CH(X)\big)\text{ for all }X\in\Linf{\CF}{}.
\end{equation}
The recursive formulation \eqref{eq:intro:tower} is often more useful than \eqref{eq:intro:strongconst} when analyzing strong consistency.  The formulation \eqref{eq:intro:tower}, however, cannot be extended in a straight forward manner to the multivariate case. Firstly, note that \eqref{eq:intro:tower} is not even well-defined in the multivariate case since $\rho_\CH(X)$ is not a $d$-dimensional random vector but a random number. Secondly, also in the univariate case the equivalence  $\eqref{eq:intro:strongconst} \Leftrightarrow \eqref{eq:intro:tower}$ only holds for risk measures that are normalized on constants, which in the {\em monetary} univariate case is implied up to a normalization by requiring that this class of risk measures satisfy cash-additivity ($\eta_\CG(X+a)=\eta_\CG(X)-a$). For multivariate risk measures there is neither a canonical extension of the concept of cash-additivity nor is it clear that such a property is desirable at all. In a first step we therefore derive a generalization of 
the 
recursive formulation \eqref{eq:intro:tower} of strong consistency for not necessarily cash-additive multivariate risk measures. Indeed, under some typical regularity assumptions, one of our first results is that two multivariate conditional risk measures $\rho_\CG$ and $\rho_\CH$  with $\CG\subseteq\CH\subseteq\CF$ are strongly consistent if and only if for all $X\in\Linf{\CF}{d}$
\begin{equation}\label{eq:intro:MultRecur}
\rho_\CG(X)=\rho_\CG\big(f_{\rho_\CH}^{-1}(\rho_\CH(X))\vecone\big),
\end{equation}
where $\vecone$ is a $d$-dimensional vector with all entries equal to $1$, and $f_{\rho_\CH}^{-1}$ is the (well-defined) inverse of the function $f_{\rho_\CH}$ associated to $\rho_\CH$ given by 
\begin{equation}\label{eq:function}
f_{\rho_\CH}:\Linf{\CH}{}\to\Linf{\CH}{}; \alpha\mapsto\rho_\CH(\alpha\vecone).
\end{equation} 
The map $f_{\rho_\CH}$ describes the risk of a system where each component is equipped with the same amount of ($\CH$-constant) cash $\alpha$. Note that if $\rho_\CH$ is a univariate risk measure that is normalized on constants then $f_{\rho_\CH}=-\id$ is minus the identity map and \eqref{eq:intro:MultRecur} reduces to \eqref{eq:intro:tower}. In this sense, for a multivariate risk measure $\rho_\CH$ the generalization of the normalization on constants property that is suited for our purposes is the requirement $f_{\rho_\CH}=-\id$. Further, we remark that one can always "normalize" a given conditional risk measure $\rho_\CH$ by putting 
\begin{equation}\label{eq:normalization}
\rhoo_\CH(X):=-f_{\rho_\CH}^{-1}\circ\rho_\CH(X).
\end{equation} 
Then  $\rhoo_\CH$ is a multivariate conditional risk measure with $f_{\rhoo_\CH}=-\id$.

After studying strong consistency for general families of multivariate conditional risk measures, we move on to give a characterization of strongly consistent multivariate conditional risk measures which are also conditionally law-invariant. In contrast to before we do not require consistency with respect to a risk measure under full information, but with respect to the initial risk measure given the trivial information $\{\emptyset,\WR\}$.
These studies were triggered by the results obtained in \cite{Follmer2014} for univariate risk measures, where it is shown that the only family of univariate, strongly consistent, conditional, cash-additive, convex risk measures is the family of conditional entropic risk measures, i.e.~the conditional risk measures are conditional certainty equivalents of the form
$$\rho_\CH(X)=-u^{-1}\left(\BEW{\PW}{u(X)}{\CH}\right), \quad X\in\Linf{\CF}{}, $$
with deterministic utility function $u(x)=a+be^{\beta x}$ or $u(x)=a+bx,$ where $a\in\R$ and $b,\beta>0$ are constants. We also remark that \cite{Kupper2009} showed this characterization for the case of dynamic risk measures by an alternative proof. 
In the multivariate case we will see that every strongly consistent family of multivariate conditionally law-invariant conditional risk measures consists of risk measures of type
\begin{equation}\label{eq:StochCertEq}
\rho_\CH(X)=f_{\rho_\CH}\left(f_{u}^{-1}\big(\BEW{\PW}{u(X)}{\CH}\big)\right),\quad X\in\Linf{\CF}{d},
\end{equation} where $u:\R^d\to \R$ is a multivariate utility function and $f_u(x):=u(x\vecone)$, $x\in \R$.
In other words they can be decomposed into the function $f_{\rho_\CH}$ in \eqref{eq:function} applied to a multivariate conditional certainty equivalent $\left(f_{u}^{-1}\big(\BEW{\PW}{u(X)}{\CH}\big)\right)$. For the study of univariate conditional certainty equivalents and their dynamic behavior we refer the interested reader to \cite{Frittelli2011b}. Moreover, we will derive the decomposition \eqref{eq:aim2} from  \eqref{eq:StochCertEq}, i.e.\ in terms of $u$ and $f_u$.

\paragraph{Structure of the paper}~

\smallskip\noindent
In Section 2 we introduce our notation and multivariate conditional risk measures. Moreover, we give the definition and some auxiliary results for the function $f_{\rho_\CH}$ mentioned in \eqref{eq:function}. In Sections~3 and 4 we prove our main results outlined above for two strongly consistent conditional risk measures, where the law-invariant case is studied in Section~4. Throughout Section 5 we extend these results to families of multivariate conditional risk measures.   

\section{Definitions and basic results}
Throughout this paper $(\WR,\CF,\PW)$ is a probability space. 
For $d\in\N$ we denote by $\Linf{\CF}{d}:=\{X=(X_1,...,X_d):X_i\in L^\infty(\WR,\CF,\PW)\;\forall i\}$ the space of equivalence classes of $\CF$-measurable, $\PW$-almost surely (a.s.) bounded random vectors. It is a Banach space when equipped with the norm $\|X\|_{d,\infty}:=\max_{i=1,\ldots,d}\supnorm{X_i}$ where $\|F\|_\infty:=\operatorname{esssup} |F|$ is the supremum norm for  $F\in L^\infty(\WR,\CF,\PW)$. We will use the usual componentwise orderings on $\R^d$ and $\Linf{\CF}{d}$, i.e.\ $x=(x_1,\ldots,x_d)\geq y=(y_1,\ldots,y_d)$ for $x,y\in \R^d$ if and only if $x_i\geq y_i$ for all $i=1,\ldots, d$, and similarly $\eqX\geq\eqY$ if and only if $\eqX_i\geq \eqY_i$ $\PW$-a.s.\ for all $i=1,...,d$.
Furthermore $\vecone$ and $\veczero$ denote the $d$-dimensional vectors whose entries are all equal to 1 or all equal to 0, respectively. 

\begin{Definition}\th\label{def:rm}
Let $\CG\subseteq\CF$. A conditional risk measure (CRM) is a function 
$$\rho_\CG:\Linf{\CF}{d}\to\Linf{\CG}{},$$
possessing the following properties:
\begin{itemize}
\item[i)] There exists a position with zero risk, i.e.\  $0\in\Imag \rho_\CG$.
\item[ii)] {\bf Strict Antitonicity:} $X\geq Y$ and $\PW(X>Y)>0$ implies  $\rho_\CG(X)\leq\rho_\CG(Y)$ and $\PW\big(\rho_\CG(X)<\rho_\CG(Y)\big)>0$.
\item[iii)] {\bf $\CG$-Locality:} For all $A\in\CG$ we have $\rho_\CG(X\ind_A+Y\ind_{A^C})=\rho_\CG(X)\ind_A+\rho_\CG(Y)\ind_{A^C}$.
\item[iv)] {\bf Lebesgue property:} If $(X_n)_{n\in\N}\subset \Linf{\CF}{d}$ is a $\|\cdot\|_{d,\infty}$-bounded sequence such that $X_n\to X$ $\PW$-a.s., then
$$\rho_\CG(X)=\lim_{n\to\infty}\rho_\CG(X_n)\quad\PW\text{-a.s.}$$
\end{itemize} 
\end{Definition}

We remark that the properties in \th\ref{def:rm} are standard in the literature on conditional risk measures.
Note that strict antitonicity is sometimes also referred to as strong sensitivity in the literature.
In order to stress the dimension we often use the term univariate conditional risk measure for a conditional risk measure as defined in Definition~\ref{def:rm} with $d=1$ and we typically denote it by $\eta_\CG$. For $d>1$ the risk measure $\rho_\CG$ of  Definition~\ref{def:rm} is called multivariate conditional risk measure.


%
%
A standard assumption on univariate CRMs is cash-additivity, i.e.\ $\eta_\CG(X+\alpha)=\eta_\CG(X)-\alpha$ for all $\alpha\in L^\infty(\CG)$, which in particular implies that we postulate a certain behavior of the risk measure $\eta_\CG$ on ($\CG$)-constants $\alpha\in L^\infty(\CG)$ which turns out to be helpful in the study of consistency. Since we do not require this property - given that a multivariate analogue is tricky to define and probably not reasonable to ask for - 
we will have to extract the behavior of a CRM on constants in the following way.

\begin{Definition}\th\label{def:f} 
For every CRM $\rho_\CG$ we introduce the function
$$f_{\rho_\CG}:\Linf{\CG}{}\to\Linf{\CG}{};\alpha\mapsto\rho_\CG(\alpha\vecone)$$
and the corresponding inverse function 
$$f^{-1}_{\rho_\CG}:\Imag f_{\rho_\CG}\to\Linf{\CG}{};\beta\mapsto \alpha\text{ such that }f_{\rho_\CG}(\alpha)=\beta.$$
\end{Definition}

\begin{Remark}\th\label{r:1}
Note that the strict antitonicity of $\rho_\CG$ implies that the inverse function $f_{\rho_\CG}^{-1}$ in \thref{def:f} is well-defined. Indeed let $\beta\in\Imag f_{\rho_\CG}$ and $\alpha_1,\alpha_2\in\Linf{\CG}{}$ such that $f_{\rho_\CG}(\alpha_1)=\beta=f_{\rho_\CG}(\alpha_2)$. Suppose that $\PW(A)>0$ where $A:=\{\alpha_1>\alpha_2\}\in\CG$. Then by strict antitonicity and $\CG$-locality we obtain that
\begin{align*}
\beta\ind_A+\rho_\CG(\veczero)\ind_{A^C}&=\rho_\CG(\alpha_1\vecone)\ind_A+\rho_\CG(\veczero)\ind_{A^C}=\rho_\CG(\alpha_1\vecone\ind_A)\\
&\leq\rho_\CG(\alpha_2\vecone\ind_A)=\rho_\CG(\alpha_2\vecone)\ind_A+\rho_\CG(\veczero)\ind_{A^C}\\
&=\beta\ind_A+\rho_\CG(\veczero)\ind_{A^C},
\end{align*}
and the inequality is strict with positive probability which is a contradiction. Thus we have that $\PW(\alpha_1>\alpha_2)=0$. The same argument for $\{\alpha_1<\alpha_2\}$ yields  $\alpha_1=\alpha_2$ $\PW$-a.s.
\end{Remark}

Next we will show that properties of $\rho_\CG$ transfer to $f_{\rho_\CG}$ and $f^{-1}_{\rho_\CG}$. Since the domain of $f_{\rho_\CG}^{-1}$ might be only a subset of $\Linf{\CG}{}$, we need to adapt the definition of the Lebesgue property for $f^{-1}_{\rho_\CG}$ in the following way:
If $(\beta_n)_{n\in\N}\subset \Imag f_{\rho_\CG}$ is a sequence which is lower- and upper-bounded by some $\underline{\beta},\overline{\beta}\in\Imag f_{\rho_\CG}$, i.e.\ $\underline{\beta}\leq\beta_n\leq\overline\beta$ for all $n\in\N$, and such that $\beta_n\to\beta$ $\PW$-a.s., then $f_{\rho_\CG}^{-1}(\beta_n)\to f_{\rho_\CG}^{-1}(\beta)$ $\PW$-a.s.
Note that this alternative definition of the Lebesgue property is equivalent to  \thref{def:rm}~(iv) if the domain is $\Linf{\CG}{}$. The properties 'strict antitonicity' and 'locality' of $f_{\rho_\CG}$ or $f^{-1}_{\rho_\CG}$ are defined analogous to \thref{def:rm}~(ii) and (iii).

\begin{Lemma}\th\label{l:-1}
Let $f_{\rho_\CG}$ and $f^{-1}_{\rho_\CG}$ be as in \thref{def:f}. Then $f_{\rho_\CG}$ and  $f^{-1}_{\rho_\CG}$ are strictly antitone, $\CG$-local and fulfill the Lebesgue property.
\end{Lemma}

\begin{proof}
For $f_{\rho_\CG}$ the statement follows immediately from the definition and the corresponding properties of $\rho_\CG$. Concerning the properties of $f^{-1}_{\rho_\CG}$, we start by proving strict antitonicity. Let $\beta_1,\beta_2\in \Imag f_{\rho_\CG}$ such that $\beta_1\geq\beta_2$ and $\PW(\beta_1>\beta_2)>0$. Suppose that $\PW(A)>0$ where $A:=\left\{f^{-1}_{\rho_\CG}(\beta_1)>f^{-1}_{\rho_\CG}(\beta_2)\right\}\in \CG$. Then
\begin{align*}
\beta_1\ind_A+f_{\rho_\CG}(0)\ind_{A^C}&=f_{\rho_\CG}\left(f^{-1}_{\rho_\CG}(\beta_1)\right)\ind_A+f_{\rho_\CG}(0)\ind_{A^C}=f_{\rho_\CG}\left(f^{-1}_{\rho_\CG}(\beta_1)\ind_A\right)\\
&\leq f_{\rho_\CG}\left(f^{-1}_{\rho_\CG}(\beta_2)\ind_A\right)=\beta_2\ind_A+f_{\rho_\CG}(0)\ind_{A^C},
\end{align*}
and the inequality is strict on a set with positive probability since $f_{\rho_\CG}$ is strictly antitone. This of course contradicts  $\beta_1\geq\beta_2$. Hence $f^{-1}_{\rho_\CG}(\beta_1)\leq f^{-1}_{\rho_\CG}(\beta_2)$. Moreover, as 
$$\PW(\beta_1>\beta_2)=\PW\left(f_{\rho_\CG}\left(f^{-1}_{\rho_\CG}(\beta_1)\right)>f_{\rho_\CG}\left(f^{-1}_{\rho_\CG}(\beta_2)\right)\right)>0$$
we must have $f^{-1}_{\rho_\CG}(\beta_1)\neq f^{-1}_{\rho_\CG}(\beta_2)$ with positive probability, i.e.\
$$\PW\left(f^{-1}_{\rho_\CG}(\beta_1)<f^{-1}_{\rho_\CG}(\beta_2)\right)>0.$$
Now we show that $f^{-1}_{\rho_\CG}$ is $\CG$-local. Let $\beta_1,\beta_2\in\Imag f_{\rho_\CG}$ as well as $A\in\CG$ be arbitrary. Further let $\alpha_i=f^{-1}_{\rho_\CG}(\beta_i),i=1,2$, i.e.\ $f_{\rho_\CG}(\alpha_i)=\beta_i$. Then we have that
$$f_{\rho_\CG}(\alpha_1\ind_A+\alpha_2\ind_{A^C})=f_{\rho_\CG}(\alpha_1)\ind_A+f_{\rho_\CG}(\alpha_2)\ind_{A^C}=\beta_1\ind_A+\beta_2\ind_{A^C}.$$
Thus $f^{-1}_{\rho_\CG}(\beta_1\ind_A+\beta_2\ind_{A^C})=\alpha_1\ind_A+\alpha_2\ind_{A^C}.$\\

\noindent Finally for the Lebesgue property let $\underline{\beta},\overline{\beta}\in\Imag f_{\rho_\CG}$ and let $(\beta_n)_{n\in\N}\subset \Imag f_{\rho_\CG}$ be a sequence with $\underline{\beta}\leq\beta_n\leq\overline{\beta}$ for all $n\in\N$ and $\beta_n\to\beta$ $\PW$-a.s.  
Consider the bounded sequences $\beta^u_n:=\sup_{k\geq n}\beta_k$ and $\beta^d_n:=\inf_{k\geq n}\beta_k$, $n\in \N$ which converge monotonically almost surely to $\beta$, i.e.\ $\beta^u_n\downarrow \beta$ $\PW$-a.s.\ and $\beta^d_n\uparrow \beta$ $\PW$-a.s.
Since $\underline\beta\leq\beta_n^u\leq \overline\beta$ for all $n\in\N$ which by antitonicity of $f^{-1}_{\rho_\CG}$ yields $f^{-1}_{\rho_\CG}(\overline\beta)\leq f^{-1}_{\rho_\CG}(\beta_n^u) \leq f^{-1}_{\rho_\CG}(\underline\beta)$, we observe that the sequence $\left(f^{-1}_{\rho_\CG}(\beta_n^u)\right)_{n\in\N}$ is uniformly bounded in $\Linf{\CG}{}$. Note that by the same argumentation also the sequences $\left(f^{-1}_{\rho_\CG}(\beta_n^d)\right)_{n\in\N}$ and $\left(f^{-1}_{\rho_\CG}(\beta_n)\right)_{n\in\N}$ are uniformly bounded in $\Linf{\CG}{}$.
 Next we will show that $\beta^u_n\in\Imag f_{\rho_\CG} $ for all $n\in\N$. Fix $n\in\N$ and set recursively
$$A_{n-1}^n:=\{\beta^u_n=\beta\}\quad\text{and}\quad A_{k}^n:=\{\beta^u_n=\beta_k\}\backslash \bigcup_{i=n-1}^{k-1}A_{i}^n,\;k\geq n,$$
then it follows from induction that $A_k^n\in\CG,k\geq{n-1}$.
Since $\sup\left\{\beta,\beta_k:k\geq n \right\}=\max\left\{\beta,\beta_k:k\geq n \right\}$, we have that $\left(\bigcup_{k\geq n-1}A_k^n\right)^C$ is a $\PW$-nullset. It follows from $\CG$-locality and the Lebesgue property of $f_{\rho_\CG}$ that
\begin{align*}
&f_{\rho_\CG}\left(f^{-1}_{\rho_\CG}(\beta)\ind_{A_{n-1}^n}+\sum_{k\geq n}f^{-1}_{\rho_\CG}(\beta_k)\ind_{A_k^n}\right)\\
&\hspace{1cm}=\beta \ind_{A_{n-1}^n}+f_{\rho_\CG}\left(\lim_{m\to\infty}\sum_{k=n}^m f_{\rho_\CG}^{-1}(\beta_k)\ind_{A_k^n}\right)\ind_{\bigcup_{k\geq n}A_{k}^n}\\
&\hspace{1cm}=\beta \ind_{A_{n-1}^n}+\lim_{m\to\infty}\left(\sum_{k=n}^m\beta_k\ind_{A_{k}^n}+f_{\rho_\CG}\left(0\right)\ind_{\bigcup_{k\geq m}A_k^n}\right)\\
	&\hspace{1cm}=\beta \ind_{A_{n-1}^n}+\sum_{k\geq n}\beta_k\ind_{A_{k}^n}=\beta^u_n,
\end{align*}
which implies $\beta_n^u\in\Imag f_{\rho_\CG}$. By a similar argumentation we obtain $\beta_n^d\in \Imag f_{\rho_\CG}$.
Recall that $\beta_n^u\downarrow\beta$ $\PW$-a.s.\ which by antitonicity of $f_{\rho_\CG}^{-1}$ implies that the sequence $\left(f^{-1}_{\rho_\CG}(\beta_n^u)\right)_{n\in\N}$ is isotone and thus $\alpha=\lim_{n\to\infty}f_{\rho_\CG}^{-1}(\beta_n^u)$ exists in $L^\infty(\CG)$. It follows from antitonicity and the Lebesgue property of $f_{\rho_\CG}$ that
$$\beta=\lim_{n\to\infty}\beta_n^u=\lim_{n\to\infty}f_{\rho_\CG}\left(f_{\rho_\CG}^{-1}(\beta_n^u)\right)=f_{\rho_\CG}(\alpha),$$
and hence that indeed $\alpha=f_{\rho_\CG}^{-1}(\beta)$. Analogously, we obtain that $f_{\rho_\CG}(\hat \alpha)=\beta$ for $\hat\alpha=\lim_{n\to\infty}f_{\rho_\CG}^{-1}(\beta_n^d)$, and thus $\hat\alpha=\alpha=f_{\rho_\CG}^{-1}(\beta)$.
Hence, by antitonicity of $f^{-1}_{\rho_\CG}$ 
\begin{align*}f_{\rho_\CG}^{-1}(\beta)&=\lim_{n\to \infty}f_{\rho_\CG}^{-1}(\beta^u_n)\leq \liminf_{n\to \infty}f_{\rho_\CG}^{-1}(\beta_n)\\ &\leq \limsup_{n\to \infty}f_{\rho_\CG}^{-1}(\beta_n)\leq \lim_{n\to \infty}f_{\rho_\CG}^{-1}(\beta^d_n)= f_{\rho_\CG}^{-1}(\beta),
\end{align*}
 so $\lim_{n\to \infty}f_{\rho_\CG}^{-1}(\beta_n)= f_{\rho_\CG}^{-1}(\beta)$, i.e.\ $f^{-1}_{\rho_\CG}$ has the Lebesgue property.
\end{proof}

\noindent An important observation that will be needed later on is that the domain of $ f^{-1}_{\rho_\CG}$ is equal to the image of $\rho_\CG$, i.e.\ $ f^{-1}_{\rho_\CG}(\rho_\CG(X))$ is well-defined for all $X\in\Linf{\CF}{d}$.

\begin{Lemma}\th\label{l:0}
For a CRM $\rho_\CG:\Linf{\CF}{d}\to \Linf{\CG}{}$ it holds that
$$\rho_\CG(\Linf{\CF}{d})= f_{\rho_\CG}(\Linf{\CG}{}).$$
\end{Lemma}

\begin{proof}
Clearly, $\rho_\CG(\Linf{\CF}{d})\supseteq f_{\rho_\CG}(\Linf{\CG}{})$.\\
For the reverse inclusion let $X\in\Linf{\CF}{d}$. Our aim is to show that there exists an $\alpha^\ast\in\Linf{\CG}{}$ such that
\begin{equation}\label{eq:prob}\rho_\CG(X)= f_{\rho_\CG}(\alpha^\ast).\end{equation}
Define
$$P:=\left\{\alpha\in\Linf{\CG}{}\,: f_{\rho_\CG}(\alpha) \geq \rho_\CG(X) \right\}.$$
As $-\|X\|_{d,\infty}\vecone\leq X\leq \|X\|_{d,\infty}\vecone$ we have that $-\|X\|_{d,\infty}\in P$, so $P\neq\emptyset$. Moreover, $P$ is bounded from above by  $\|X\|_{d,\infty}$ since if $A:=\{\alpha>\|X\|_{d,\infty}\}$ for $\alpha\in \Linf{\CG}{}$ has positive probability, then by $\CG$-locality and strict antitonicity $$f_{\rho_\CG}(\alpha)\ind_A= f_{\rho_\CG}(\alpha \ind_A)\ind_A\leq f_{\rho_\CG}(\|X\|_{d,\infty}\ind_A)\ind_A=f_{\rho_\CG}(\|X\|_{d,\infty})\ind_A\leq \rho_\CG(X)\ind_A$$ where the first inequality is strict with positive probability, so $\alpha\not\in P$. By $\CG$-locality it also follows that $P$ is upwards directed. Hence, for $$\alpha^\ast:=\esssup P$$ there is a uniformly bounded sequence $(\alpha_n)_{n\in \N}\subset P$ such that $\alpha^\ast=\lim_{n\to \infty}\alpha_n$ $\PW$-a.s.; see \cite{Follmer2011}~Theorem~A.33. Thus it follows that $\alpha^\ast\in\Linf{\CG}{}$ and $$f_{\rho_\CG}(\alpha^\ast)=\lim_{n\to\infty}f_{\rho_\CG}(\alpha_n)\geq \rho_\CG(X),$$
i.e.\ $\alpha^\ast\in P$.
Let $$B:=\{f_{\rho_\CG}(\alpha^\ast)>\rho_\CG(X)\}$$
and note that by the Lebesgue property $$B=\bigcup_{n\in \N}\{f_{\rho_\CG}(\alpha^\ast+1/n)>\rho_\CG(X)\}\quad \mbox{$\PW$-a.s.}$$ Hence, if $\PW(B)>0$ it follows that $\PW(B_n)>0$ for some $B_n:=\{f_{\rho_\CG}(\alpha^\ast+1/n)>\rho_\CG(X)\}$. Note that $B_n\in \CG$ and that $$\alpha^\ast\ind_{B_n^C}+(\alpha^\ast+1/n)\ind_{B_n}\in P$$ by $\CG$-locality of $f_{\rho_\CG}$. But this contradicts the definition of $\alpha^\ast$. Hence, $\PW(B)=0$.
\end{proof}

Sometimes it will be useful to normalize the CRM in the following sense:

\begin{Definition}\th\label{r:3}
	We call a CRM $\rho_\CG:\Linf{\CF}{d}\to L^\infty(\CG) $ \emph{normalized on constants} if 
	$$f_{\rho_\CG}(\alpha)=-\alpha\quad\text{for all }\alpha\in\Linf{\CG}{}.$$
\end{Definition}

Indeed let $\rho_\CG:\Linf{\CF}{d}\to L^\infty(\CG)$ be a CRM and define $\rhoo_\CG:=-f_{\rho_\CG}^{-1}\circ\rho_\CG$. Then $\rhoo_\CG$ is a CRM (\thref{l:-1} and \thref{l:0}) which is normalized in the sense of being normalized on constants as defined above. We call $\rhoo_\CG$ the \emph{normalized CRM} of $\rho_\CG$.

\section{Strong consistency}
In this section we study consistency of CRMs. We consider the most frequently used consistency condition for univariate risk measures in the literature which is known as strong consistency and extend it to the multivariate case. We refer to \cite{Detlefsen2005}, \cite{Cheridito2006}, \cite{Cheridito2011}, \cite{Kupper2009}, and \cite{Penner2007} for more information on strong consistency of univariate risk measures. 
\cite{Kromer2014} also study a kind of consistency for multivariate risk measures, however, as we will point out in  \thref{r:Kromer1} below, their definition of consistency differs from our approach.
For the remainder of this section we let $\CG$ and $\CH$ be two sub-$\sigma$-algebras of $\CF$ such that $\CG\subseteq\CH$, and let $\rho_\CG:\Linf{\CF}{d}\to \Linf{\CG}{}$ and $\rho_\CH:\Linf{\CF}{d}\to \Linf{\CH}{}$ be the corresponding CRMs.

\begin{Definition}[Strong consistency]\th\label{def:consistency}
The pair $\{\rho_\CG, \rho_\CH\}$ is called \emph{strongly consistent} if 
\begin{equation}\label{eq:consistency}
\rho_\CH(X)\leq\rho_\CH(Y)\;\Rightarrow\; \rho_\CG(X)\leq \rho_\CG(Y)\quad(X,Y\in\Linf{\CF}{d}).
\end{equation}
\end{Definition}

Strong consistency states that if one risk is preferred to another risk in almost surely all states under more information, then this preference already holds under less information. 
Our first result shows that strong consistency can be equivalently defined by a recursive relation.

\begin{Lemma}\th\label{l:1}
Equivalent are:
\begin{enumerate}
	\item $\{\rho_\CG,\rho_\CH\}$ is strongly consistent;
	\item For all $X\in\Linf{\CF}{d}$ it holds that 
	$$\rho_\CG(X)=\rho_\CG\Big(f_{\rho_\CH}^{-1}\big(\rho_\CH(X)\big)\vecone\Big),$$
	where $f_{\rho_\CH}^{-1}$ was defined in \thref{def:f}.
\end{enumerate}
\end{Lemma}

\begin{proof}
\underline{(i)$\Rightarrow$(ii):}
As for all  $X\in\Linf{\CF}{d}$
$$\rho_\CH(X)=\rho_\CH\left(f_{\rho_\CH}^{-1}\left(\rho_\CH(X)\right)\vecone\right),$$
it follows from strong consistency that
$$\rho_\CG(X)=\rho_\CG\left(f_{\rho_\CH}^{-1}\left(\rho_\CH(X)\right)\vecone\right).$$
\underline{(ii)$\Rightarrow$(i):}
Let $X,Y\in\Linf{\CF}{d}$ be such that $\rho_{\CH}(X)\leq \rho_{\CH}(Y)$. Then by antitonicity of $f_{\rho_\CH}^{-1}$ and $\rho_{\CG}$ it follows that
$$\rho_{\CG}(X)=\rho_{\CG}\Big(f_{\rho_\CH}^{-1}\big(\rho_{\CH}(X)\big)\vecone\Big)\leq\rho_{\CG}\Big(f_{\rho_\CH}^{-1}\big(\rho_{\CH}(Y)\big)\vecone\Big)=\rho_{\CG}(Y).$$
\end{proof}

\begin{Remark}\th\label{r:2}
Let $\eta_\CG$ and $\eta_\CH$ be two univariate CRMs, where $\eta_\CH$ is normalized on constants, i.e.\ $\eta_\CH(\alpha)=-\alpha$ for all $\alpha\in\Linf{\CH}{}$.
Then $f_{\eta_\CH}(\alpha)=f_{\eta_\CH}^{-1}(\alpha)=-\alpha$ and thus strong consistency is equivalent to 
$$\eta_\CG(F)=\eta_\CG\big(-\eta_\CH(F)\big),\quad F\in\Linf{\CF}{}.$$
\end{Remark}

\begin{Remark}\th\label{r:34}
If $\{\rho_\CG, \rho_\CH\}$ is strongly consistent so is the pair of normalized CRMs $\{\rhoo_\CG,\rhoo_\CH\}$ as defined in \th\ref{r:3} and vice versa. Since $f_{\rhoo_\CG}=f_{\rhoo_\CH}=-\id$ strong consistency of the normalized CRMs is equivalent to 
$$\rhoo_\CG(F)=\rhoo_\CG\big(-\rhoo_\CH(F)\vecone\big),\quad F\in\Linf{\CF}{},$$
in analogy to \th\ref{r:2}.
\end{Remark}

In the following lemma we will show that strong consistency of $\{\rho_\CG, \rho_\CH\}$ uniquely determines the normalized CRM $\rhoo_\CH$. 

\begin{Lemma}\th\label{l:2}
If $\{\rho_\CG,\rho_\CH\}$ is strongly consistent, then $\rho_\CG$ uniquely determines the normalized CRM $\rhoo_\CH=-f_{\rho_\CH}^{-1}\circ\rho_\CH$. 
\end{Lemma}

\begin{proof}
Suppose that there are two CRMs $\rho_\CH^1$ and $\rho_\CH^2$ which are strongly consistent with respect to $\rho_\CG$, i.e. 
$$\rho_{\CG}\Big(f_{\rho_{\CH}^1}^{-1}\big(\rho_{\CH}^1(X)\big)\vecone\Big)=\rho_{\CG}(X)= \rho_{\CG}\Big(f_{\rho_{\CH}^2}^{-1}\big(\rho_{\CH}^2(X)\big)\vecone\Big),\quad  X\in\Linf{\CF}{d}.$$
We will show that $f_{\rho_{\CH}^1}^{-1}\left(\rho_{\CH}^1(X)\right)=f_{\rho_{\CH}^2}^{-1}\left(\rho_{\CH}^2(X)\right)$. Suppose that there exists an $X\in\Linf{\CF}{d}$ such that $A:=\left\{f_{\rho_{\CH}^1}^{-1}\left(\rho_{\CH}^1(X)\right)>f_{\rho_{\CH}^2}^{-1}\left(\rho_{\CH}^2(X)\right)\right\}\in\CH$ has positive probability.
Then, by the $\CH$-locality of $\rho_{\CH}^1$ and $\rho_{\CH}^2$, we obtain
\begin{align}
\rho_{\CG}(X\ind_A)&=\rho_{\CG}\Big(f_{\rho_{\CH}^1}^{-1}\big(\rho_{\CH}^1(X\ind_A)\big)\vecone\Big)=\rho_{\CG}\Big(f_{\rho_{\CH}^1}^{-1}\big(\rho_{\CH}^1(X)\big)\ind_A\vecone\Big)\notag\\
&\leq\rho_{\CG}\Big(f_{\rho_{\CH}^2}^{-1}\big(\rho_{\CH}^2(X)\big)\ind_A\vecone\Big)=\rho_{\CG}\Big(f_{\rho_{\CH}^2}^{-1}\big(\rho_{\CH}^2(X\ind_A)\big)\vecone\Big)\notag\\
&=\rho_{\CG}(X\ind_A).\label{eq:123}
\end{align}
where the inequality \eqref{eq:123} is strict with positive probability as $\rho_\CG$ is strictly antitone, and hence we have a contradiction. Reverting the role of $\rho_\CH^1$ and $\rho_\CH^2$ in the definition of $A$ proves the lemma.
\end{proof}

In \cite{Hoffmann2016} we studied under which conditions a (multivariate) conditional risk measure can be decomposed as in \eqref{eq:aim2}, i.e.\ into a conditional aggregation function and a univariate conditional risk measure. We will pursue showing that strong consistency of $\{\rho_\CG,\rho_\CF\}$ is already sufficient to guarantee a decomposition \eqref{eq:aim2} for both $\rho_\CG$ and $\rho_\CF$. To this end we need to clarify what we mean by a conditional aggregation function:

\begin{Definition}\th\label{def:CAF}
We call a function $\Lambda:\Linf{\CF}{d}\to\Linf{\CF}{}$ a \emph{conditional aggregation function} if it fulfills the following properties:
\begin{description}
\item[Strict isotonicity:] $X\geq Y$ and $\PW(X>Y)>0$ implies $\Lambda(X)\geq\Lambda(Y)$ and $\PW\big(\Lambda(X)>\Lambda(Y)\big)>0$.
\item[$\CF$-Locality:] $\Lambda(X\ind_A+Y\ind_{A^C})=\Lambda(X)\ind_A+\Lambda(Y)\ind_{A^C}$ for all $A\in\CF$;
\item[Lebesgue property:] For any uniformly bounded sequence $(X_n)_{n\in\N}$ in $\Linf{\CF}{d}$ such that $X_n\to X$ $\PW$-a.s., we have that
$$\Lambda(X)=\lim_{n\to\infty}\Lambda(X_n)\quad\PW\text{-a.s.}$$
\end{description}  
Moreover for $\CH\subset\CF$, we call $\Lambda$ a \emph{$\CH$-conditional aggregation function} if in addition $$\Lambda(\Linf{\CJ}{d})\subseteq\Linf{\CJ}{}\text{ for all }\CH\subseteq\CJ\subseteq\CF.$$
\end{Definition}

\begin{Remark}
The name $\CH$-conditional aggregation function refers to the fact that $\Lambda(x)\in\Linf{\CH}{}$ for all $x\in\R^d$. Thus every conditional aggregation function is at least a $\CF$-conditional aggregation function. 
 \end{Remark}

As for conditional risk measures we define: 

\begin{Definition}\th\label{def:fLambda}
For a conditional aggregation function $\Lambda:\Linf{\CF}{d}\to\Linf{\CF}{}$ let 
$$f_{\Lambda}:\Linf{\CF}{}\to\Linf{\CF}{};F\mapsto\Lambda(F\vecone)$$
and
$$f^{-1}_{\Lambda}:\Imag f_{\Lambda}\to\Linf{\CF}{};G\mapsto F\text{ such that }f_{\Lambda}(F)=G.$$
\end{Definition}

\begin{Lemma}\th\label{r:fLambdaprop}
Let $\Lambda:\Linf{\CF}{d}\to\Linf{\CF}{}$ be a conditional aggregation function. Then  $f_{\Lambda}$ and $f_{\Lambda}^{-1}$ are strictly isotone, $\CF$-local, and fulfill the Lebesgue property. Moreover, $\Lambda(\Linf{\CF}{d})=f_{\Lambda}(\Linf{\CF}{})$ and  $\Lambda(X)=\Lambda\big(f_{\Lambda}^{-1}(\Lambda(X))\vecone\big)$ for all $X\in\Linf{\CF}{d}$.
\end{Lemma}

\noindent
The well-definedness of $f_{\Lambda}^{-1}$ follows as in \th\ref{r:1}. Further the proof of \thref{r:fLambdaprop} is analogous to the proofs of \thref{l:-1} and \thref{l:0} and therefore omitted here. 

In order to state the decomposition result for strongly consistent CRMs, we first recall the main result from \cite{Hoffmann2016} adapted to the framework of this paper in \thref{T1} for which we need the following definition.

\begin{Definition}\th\label{def:techreal}
We say that a function $\rho_\CG: \Linf{\CF}{d}\to \Linf{\CG}{}$ has a continuous realization $\rho_\CG(\cdot,\cdot)$, if for all $X\in\Linf{\CF}{d}$ there exists a representative $\rho_\CG(X,\cdot)$ of the equivalence class $\rho_\CG(X)$ such that $\Wrho_\CG:\R^d\times \WR\to\R; (x,\om)\mapsto\rho_\CG(\eqx,\om)$ is continuous in its first argument $\PW$-a.s.
\end{Definition}


\begin{Proposition}\th\label{T1}
Let $\rho_\CG: \Linf{\CF}{d}\to\Linf{\CG}{}$ be a CRM and suppose that there exists a continuous realization $\rho_\CG(\cdot,\cdot)$ which satisfies \emph{risk-antitonicity}: $$\mbox{  
$\Wrho_\CG(\eqXom,\om)\geq\Wrho_\CG(\eqYom,\om)$ $\PW$-a.s., implies $\rho_\CG(\eqX)\geq\rho_\CG(\eqY)$.}$$  
Then there exists a $\CG$-conditional aggregation function $\Lambda_\CG:\Linf{\CF}{d}\to\Linf{\CF}{}$ and a univariate CRM $\eta_\CG: \Imag\Lambda_\CG\to \Linf{\CG}{}$ such that
$$\rho_\CG \left(\eqX\right)=\eta_\CG\left(\Lambda_\CG(\eqX)\right)\quad \text{for all $\eqX\in \Linf{\CF}{d}$}$$
and
\begin{equation}\label{eq:constonAF}
\eta_\CG\left(\Lambda_\CG(X)\right)=-\Lambda_\CG(X)\quad \mbox{for all $X\in\Linf{\CG}{d}$.}
\end{equation}
This decomposition is unique.
\end{Proposition}

\begin{proof}
Since $\rho_\CG$ is antitone, $\R^d\ni x\mapsto \rho_\CG(x)$ is antitone. It has been shown in  \cite{Hoffmann2016}~Theorem~2.10 that this property in conjunction with the fact that $\rho_\CG$ has a continuous realization which fulfills risk-antitonicity is sufficient for the existence and uniqueness of a function $\Lambda_\CG:\Linf{\CF}{d}\to\Linf{\CF}{}$ which is isotone, $\CF$-local and fulfills the Lebesgue property and a function $\eta_\CG: \Imag \Lambda_\CG\to \Linf{\CG}{}$ which is antitone such that 
\begin{equation}\label{eq:help1}
\rho_\CG=\eta_\CG\circ\Lambda_\CG\quad \mbox{and}\quad  \eta_\CG\big(\Lambda_\CG(x)\big)=-\Lambda_\CG(x)\;\mbox{for all $x\in\R^d$.}\end{equation} Note that in the proof of Theorem~2.10 in \cite{Hoffmann2016} $\Lambda_\CG$ is basically constructed by setting $\Lambda_\CG(X)(\omega)=-\Wrho_\CG(\eqXom,\om)$, which implies that $\Lambda_\CG$ is necessarily $\CF$-local even though this is not directly mentioned in the paper. Indeed in \cite{Hoffmann2016} we do not require or mention locality at all. 

It remains to be shown that $\Lambda_\CG$ is a $\CG$-conditional aggregation function, $\eta_\CG$ is a univariate CRM on $\Imag\Lambda_\CG$, and that \eqref{eq:constonAF} holds. 
First of all, we show that $\CF$-locality and \eqref{eq:help1} imply \eqref{eq:constonAF}. To this end denote by $\CS$ the set of $\CF$-measurable simple random vectors, i.e.\ $X\in\CS$ if $X$ is of the form $X=\sum_{i=1}^k x_i\ind_{A_i}$, where $k\in\N$, $x_i\in\R^d$ and $A_i\in\CF$, $i=1,...,k$, are disjoint sets such that $\PW(A_i)>0$ and $\PW(\bigcup_{i=1}^k A_i)=1$. Now let $X\in \Linf{\CG}{d}$. Pick a uniformly bounded sequence  $(X_n)_{n\in\N}=\left(\sum_{i=1}^{k_n} x^n_i\ind_{A^n_i}\right)_{n\in\N}\subset\CS$ such that $A^n_i\in \CG$ for all $i=1,\ldots, k_n$, $n\in \N$, and $X_n\to X$ $\PW$-a.s. Then by \eqref{eq:help1}, $\CF$-locality and the Lebesgue property of $\Lambda_\CG$ and $\rho_\CG$ we infer that 
\begin{eqnarray*}-\Lambda_\CG(X)&=&-\lim_{n\to \infty}\Lambda_\CG(X_n)\;=\;\lim_{n\to \infty}\sum_{i=1}^{k_n}-\Lambda_\CG(x_i^n)\ind_{A_i^n} \\ &=&\lim_{n\to \infty}\sum_{i=1}^{k_n}\rho_\CG(x_i^n)\ind_{A_i^n}\; =\; \lim_{n\to \infty} \rho_\CG(X_n) \; =\; \rho_\CG(X), 
\end{eqnarray*}
which proves \eqref{eq:constonAF}.
Next we show that $\Lambda_\CG$ is a $\CG$-conditional aggregation function. The yet missing  properties which need to be verified are strict antitonicity and that $\Lambda_\CG$ is $\CG$-conditional.  The latter follows from \cite{Hoffmann2016}~Lemma~3.1. As for strict antitonicity let $X,Y\in\Linf{\CF}{d}$ with $X\geq Y$ such that $\PW(X>Y)>0$. Then by isotonicity of $\Lambda_\CG$ we have that $\Lambda_\CG(X)\geq\Lambda_\CG(Y)$. Suppose that $\Lambda_\CG(X)=\Lambda_\CG(Y)$ $\PW$-a.s., then $$\rho_\CG(X)=\eta_\CG(\Lambda_\CG(X))= \eta_\CG(\Lambda_\CG(Y))=\rho_\CG(Y)$$ which contradicts strict antitonicity of $\rho_\CG$. 
Thus $\Lambda_\CG$ fulfills all properties of a $\CG$-conditional aggregation function.\\
As for $\eta_\CG$, note that by \thref{r:fLambdaprop} for all $F\in\Imag\Lambda_\CG$ we have that \begin{equation}\label{eq:help2}\eta_\CG(F)=\eta_\CG\big(\Lambda_\CG\big(f^{-1}_{\Lambda_\CG}(F)\vecone\big)\big)=\rho_\CG\big(f_{\Lambda_\CG}^{-1}(F)\vecone\big),\end{equation} where $f_{\Lambda_\CG}^{-1}$ was defined in \thref{def:fLambda}. Since $\rho_\CG$ and $f^{-1}_{\Lambda_\CG}$ are strictly monotone, $\CG$-local, and fulfill the Lebesgue property, so does  $\eta_\CG$, i.e.\ $\eta_\CG$ is a univariate CRM on $\Imag\Lambda_\CG$.
%
\end{proof}

\begin{Theorem}\th\label{l:riskanti2}
Let $\rho_\CG:\Linf{\CF}{d}\to \Linf{\CG}{}$ and  $\rho_\CF:\Linf{\CF}{d}\to \Linf{\CF}{}$ be CRMs such that $\{\rho_\CG,\rho_\CF\}$ is strongly consistent. Moreover, suppose that 
\begin{equation}\label{eq:const}f_{\rho_\CF}^{-1}\circ\rho_\CF(x)\in\R\quad \mbox{for all $x\in\R^d$}.\end{equation}
If $\rho_\CG$ has a continuous realization $\rho_\CG(\cdot,\cdot)$, then there exists a $\CG$-conditional aggregation function $\Lambda_\CG:\Linf{\CF}{d}\to\Linf{\CF}{}$ and a univariate CRM $\eta_\CG: \Imag\Lambda_\CG\to \Linf{\CG}{}$ such that 
\begin{equation}\label{eq:help5}\rho_\CG(X)=\eta_\CG\big(\Lambda_\CG(X)\big)\quad \text{for all }X\in\Linf{\CF}{d}\end{equation} and
$$\eta_\CG\big(\Lambda_\CG(X)\big)=-\Lambda_\CG(X)\quad \text{for all }X\in\Linf{\CG}{d}.$$
Let $\Lambda_\CF:=-\rho_\CF$ and $\eta_\CF:=-\id$ so that $\rho_\CF=\eta_\CF\circ\Lambda_\CF$ for the $\CF$-conditional aggregation function $\Lambda_\CF$ and the univariate CRM $\eta_\CF$. Then
\begin{equation}\label{eq:consoflam}
\Lambda_\CF(X)\leq\Lambda_\CF(Y)\quad\Longrightarrow\quad\Lambda_\CG(X)\leq\Lambda_\CG(Y)\quad(X,Y\in\Linf{\CF}{d}),
\end{equation}
i.e.\ $\Lambda_\CG$ and $\Lambda_\CF$ are strongly consistent. 

Conversely, suppose that the CRM $\rho_\CG:\Linf{\CF}{d}\to \Linf{\CG}{}$ satisfies \eqref{eq:help5}, then $\{\rho_\CG,\rho_\CF\}$ is strongly consistent where $\rho_\CF:=-\Lambda_\CG$ is a CRM.
\end{Theorem}

We remark that in \thref{l:riskanti2} we require consistency of the pair $\{\rho_\CG,\rho_\CF\}$ where $\rho_\CF$ is a CRM given the full information $\CF$. Note that $\rho_\CF$ is (apart from the sign) simply a conditional aggregation function as defined in \thref{def:CAF}, so $\rho_\CG$ is required to be consistent with some aggregation function under full information. This also explains $\Lambda_\CF$. For $d=1$ this consistency is automatically satisfied by monotonicity (and the aggregation is simply the identity function), and clearly the assertion is trivial anyway. For higher dimensions, \thref{l:riskanti2} states that if there exists an aggregation function which is consistent with $\rho_\CG$, then $\rho_\CG$ is  automatically of type \eqref{eq:help5}. Clearly, if we already know that \eqref{eq:help5} holds true, then $\rho_\CG$ is consistent with $\rho_\CF=-\Lambda_\CG$. Consistency with an aggregation under full information is a very natural requirement, because even under full information, so 
without 
risk, typically the losses still need to be aggregated in some way, and therefore any CRM under less information $\CG$ should respect this aggregation. 

Note also that the condition \eqref{eq:const} is a slight strengthening of being normalized on constants, the latter being automatically satisfied by the very definition of the normalization $f_{\rho_\CF}^{-1}\circ\rho_\CF$; see above. 

The following proof of \thref{l:riskanti2} is based on two observations: $\rho_\CF$ is necessarily risk-antitone as defined in \thref{T1}. Strong consistency in turn implies that risk-antitonicity of $\rho_\CF$ is passed on (backwards) to $\rho_\CG$, and hence \thref{T1} applies.

\begin{proof}[Proof of \thref{l:riskanti2}:] In case we already know that \eqref{eq:help5} holds, then by antitonicity of $\eta_\CG$ it follows that $\{\rho_\CG, -\Lambda_\CG\}$ is strongly consistent, and clearly $-\Lambda_\CG:\Linf{\CF}{d}\to \Linf{\CF}{}$ is also a CRM. Thus the last assertion of \thref{l:riskanti2} is proved.

In order to show the first part of \thref{l:riskanti2}, we recall that the only property which remains to be shown in order to apply \thref{T1} is risk-antitonicity of $\rho_\CG$:
For this purpose we first consider simple random vectors $X,Y\in \CS$ where $\CS$ was defined in the proof of \thref{T1}. Note that there is no loss of generality by assuming that $X=\sum_{i=1}^n x_i\ind_{A_i}\in\CS$ and $Y=\sum_{i=1}^n y_i\ind_{A_i}\in\CS$, i.e.\ the partition $(A_i)_{i=1,\ldots,n}$ of $\Omega$ is the same for $X$ and $Y$.  Suppose that $\Wrho_\CG(\eqXom,\om)\geq\Wrho_\CG(\eqYom,\om)$ $\PW$-a.s. It follows that
$$\Wrho_\CG(x_i,\om)\geq\Wrho_\CG(y_i,\om)\quad\text{for all }\om\in A_i\backslash N,i=1,...,n,$$ where $N$ is a $\PW$-nullset. 
We claim that this implies \begin{equation}\label{eq:help3}f_{\rho_\CG}^{-1}\big(\rho_\CG(x_i)\big)\leq f_{\rho_\CG}^{-1}\big(\rho_\CG(y_i)\big)\quad \mbox{for all $i=1,...,n$.}\end{equation} In order to verify this, we first notice that as $\rho_\CG$ and $\rho_\CF$ are strongly consistent and by \eqref{eq:const} we have for all $x\in\R^d$ that
\begin{equation}\label{eq:const21}
f_{\rho_\CG}^{-1}\big(\rho_\CG(x)\big)=f_{\rho_\CG}^{-1}\big(\rho_\CG\big(f_{\rho_\CF}^{-1}\big(\rho_\CF(x)\big)\vecone\big)\big)=f_{\rho_\CF}^{-1}\big(\rho_\CF(x)\big)\in\R. 
\end{equation} Here we also used that the normalization $-f_{\rho_\CG}^{-1}\circ \rho_\CG$ is normalized on constants. In other words $f_{\rho_\CG}^{-1}\big(\rho_\CG(x_i)\big)$
 and $f_{\rho_\CG}^{-1}\big(\rho_\CG(y_i)\big)$ are real numbers. Next we define 
$B_i:=\{\omega\in \Omega\mid \Wrho_\CG(x_i,\omega)\geq\Wrho_\CG(y_i,\omega)\}\in\CG$. Then $(A_i\setminus N)\subseteq B_i$ and hence $\PW(B_i)>0$ for all $i=1,...,n$. Using antitonicity and $\CG$-locality of $f_{\rho_\CG}^{-1}$ we obtain 
\begin{align*}
f^{-1}_{\rho_\CG}\big(\rho_\CG(x_i)\big)\ind_{B_i} &=f^{-1}_{\rho_\CG}\big(\rho_\CG(x_i)\ind_{B_i}\big)\ind_{B_i}
\leq f^{-1}_{\rho_\CG}\big(\rho_\CG(y_i)\ind_{B_i}\big)\ind_{B_i}=f^{-1}_{\rho_\CG}\big(\rho_\CG(y_i)\big)\ind_{B_i}.
\end{align*}
As $f_{\rho_\CG}^{-1}\big(\rho_\CG(y_i)\big)$ are indeed real numbers, \eqref{eq:help3} follows.

Now by strong consistency of $\{\rho_\CG,\rho_\CF\}$, $\CF$-locality of $\rho_\CF$ and $f^{-1}_{\rho_\CF}$, and by \eqref{eq:const21} as well as antitonicity of $\rho_\CG$ we obtain 
\begin{eqnarray*}
\rho_\CG(X)&=&\rho_\CG\left(f_{\rho_\CF}^{-1}\big(\rho_\CF(X)\big)\vecone\right) \quad =\quad  \rho_\CG\left( \sum_{i=1}^n f_{\rho_\CF}^{-1}\big(\rho_\CF(x_i)\big)\ind_{A_i}\vecone\right) \\
&=& \rho_\CG\left(\sum_{i=1}^n f_{\rho_\CG}^{-1}\big(\rho_\CG(x_i)\big)\ind_{A_i}\vecone\right)  \quad\geq\quad \rho_\CG\left( \sum_{i=1}^n f_{\rho_\CG}^{-1}\big(\rho_\CG(y_i)\big)\ind_{A_i}\vecone\right) \\
&=& \rho_\CG(Y),
\end{eqnarray*} which proves risk-antitonicity for simple random vectors $X,Y\in \CS$.
For general $X,Y\in\Linf{\CF}{d}$ with $\Wrho_\CG(\eqXom,\om)\geq\Wrho_\CG(\eqYom,\om)$ for $\PW$-a.e.\ $\om\in\WR$ we can find uniformly bounded sequences $(X_n)_{n\in\N},(Y_n)_{n\in\N}\subset\CS$ such that $X_n\nearrow X$ and $Y_n\searrow Y$ $\PW$-a.s.\ for $n\to\infty$. 
Then by antitonicity
$$\Wrho_\CG(X_n(\om),\om)\geq\Wrho_\CG(\eqXom,\om)\geq\Wrho_\CG(\eqYom,\om)\geq\Wrho_\CG(Y_n(\om),\om)\text{ for }\PW\text{-a.s.}$$ 
Therefore, $\rho_\CG(X_n)\geq \rho_\CG(Y_n)$  and the Lebegue property of $\rho_\CG$ yield
$$\rho_\CG(X)=\lim_{n\to\infty}\rho_\CG(X_n)\geq\lim_{n\to\infty}\rho_\CG(Y_n)=\rho_\CG(Y).$$
Thus $\rho_\CG$ is risk-antitone and we apply \thref{T1}. Hence, there is  a $\CG$-conditional aggregation function $\Lambda_\CG:\Linf{\CF}{d}\to\Linf{\CF}{}$ and a univariate CRM $\eta_\CG: \Imag \Lambda_\CG\to \Linf{\CG}{}$ such that $\rho_\CG=\eta_\CG\circ\Lambda_\CG$ and $\eta_\CG\big(\Lambda_\CG(X)\big)=-\Lambda_\CG(X)$ for all $X\in\Linf{\CG}{d}$.\\
Let $X,Y\in\Linf{\CF}{d}$ such that 
$$\Lambda_\CF(X)=-\rho_\CF(X)\leq-\rho_\CF(Y)=\Lambda_\CF(Y)$$ and let $(X_n)_{n\in\N}\subset\CS$ and $(Y_n)_{n\in\N}\subset\CS$ be uniformly bounded sequences such that $X_n\nearrow X$ and $Y_n\searrow Y$ $\PW$-a.s.\ for $n\to\infty$. Again there is no loss in assuming that both $X_n$ and $Y_n$ for given $n\in \N$ are defined over the same partition, i.e.\ $X_n=\sum_{i=1}^{k_n}x_i^n\ind_{A_i^n}$ and $Y_n=\sum_{i=1}^{k_n}y_i^n\ind_{A_i^n}$.  By the $\CF$-locality and antitonicity of $\rho_\CF$ it follows that for all $n\in\N$
$$\sum_{i=1}^{k_n}\rho_\CF(x_i^n)\ind_{A_i^n}\geq\rho_\CF(X)\geq\rho_\CF(Y)\geq\sum_{i=1}^{k_n}\rho_\CF(y_i^n)\ind_{A_i^n}.$$ As $f^{-1}_{\rho_\CF}\circ \rho_\CF(x_i^n)$ and $f^{-1}_{\rho_\CF}\circ \rho_\CF(y_i^n)$ are real numbers according to assumption \eqref{eq:const} and as the above computation shows that $f^{-1}_{\rho_\CF}\circ \rho_\CF(x_i^n)\leq f^{-1}_{\rho_\CF}\circ \rho_\CF(y_i^n)$ on $A_i^n$, we obtain, as above that indeed $f^{-1}_{\rho_\CF}\circ \rho_\CF(x_i^n)\leq f^{-1}_{\rho_\CF}\circ \rho_\CF(y_i^n)$, $i=1,\ldots, k_n$. Now strong consistency and \eqref{eq:constonAF} imply that
\begin{align*}
-\Lambda_\CG(x_i^n)&=\rho_\CG(x_i^n)=\rho_\CG(f^{-1}_{\rho_\CF}\circ \rho_\CF(x_i^n))\\
	&\geq \rho_\CG(f^{-1}_{\rho_\CF}\circ \rho_\CF(y_i^n)) = \rho_\CG(y_i^n)\\
	&= -\Lambda_\CG(y_i^n)
\end{align*} 
and hence by $\CG$-locality of $\Lambda_\CG$ 
$$\Lambda_\CG(X_n)=\sum_{i=1}^{k_n}\Lambda_\CG(x_i^n)\ind_{A_i^n}\leq\sum_{i=1}^{k_n}\Lambda_\CG(y_i^n)\ind_{A_i^n}=\Lambda_\CG(Y_n).$$
Finally we conclude with the Lebesgue property that 
$$\Lambda_\CG(X)=\lim_{n\to\infty}\Lambda_\CG(X_n)\leq\lim_{n\to\infty}\Lambda_\CG(Y_n)=\Lambda_\CG(Y).$$
\end{proof}

\begin{Remark}\th\label{r:321}
We know from \thref{r:fLambdaprop} that the inverse function $f_{\LG}^{-1}$ of $f_{\LG}$ is isotone and that $\LG(X)=\LG\big(f_{\LG}^{-1}(\LG(X))\vecone\big)$ for all $X\in\Linf{\CF}{d}$.
Therefore it can be shown as in \thref{l:1}, that \eqref{eq:consoflam} is equivalent to 
$$f_{\Lambda_\CG}^{-1}\big(\Lambda_\CG(X)\big)=f_{\Lambda_\CF}^{-1}\big(\Lambda_\CF(X)\big),\text{ for all }X\in\Linf{\CF}{d}.$$
Note that we cannot write the recursive form of the strong consistency of two CRMs $\rho_\CG$ and $\rho_\CF$ as above, since $f_{\rho_\CG}$ is only defined on $\Linf{\CG}{}$ and not on $\Linf{\CF}{}$ in constrast to $f_{\Lambda_\CG}$.
\end{Remark}

In the following Theorem we summarize our findings from \thref{T1} and \thref{l:riskanti2} on CRMs which extend the results in \cite{Hoffmann2016} for strong consistency:

\begin{Theorem}\th\label{t:triangle}
If $\rho_\CG:\Linf{\CF}{d}\to\Linf{\CG}{}$ is a CRM with a continuous realization $\rho_\CG(\cdot,\cdot)$ and satisfies $f_{\rho_\CG}^{-1}\circ \rho_\CG(x)\in\R$ for all $x\in\R^d$, then the following three statements are equivalent
\begin{enumerate}
\item\label{it:1} $\rho_\CG(\cdot,\cdot)$ is risk-antitone;
\item\label{it:2} $\rho_\CG$ is decomposable as in \eqref{eq:help5};
\item\label{it:3} $\rho_\CG$ is strongly consistent with some aggregation function $\Lambda:\Linf{\CF}{d}\to\Linf{\CF}{}$, i.e.\ $\{\rho_\CG,-\Lambda\}$ is strongly consistent.
\end{enumerate}
\end{Theorem}

\begin{proof}
The equivalence of \eqref{it:2} and \eqref{it:3} has been shown in \thref{l:riskanti2} and that \eqref{it:1} implies \eqref{it:2} follows from \thref{T1}. Finally, the proof of \thref{l:riskanti2} shows that \eqref{it:3} implies \eqref{it:1}.
	
	
\end{proof}
\section{Conditional law-invariance and strong consistency}

As in the previous section, if not otherwise stated, throughout this section we let $\CG$ and $\CH$ be two sub-$\sigma$-algebras of $\CF$ such that $\CG\subseteq\CH$, and let $\rho_\CG:\Linf{\CF}{d}\to \Linf{\CG}{}$ and $\rho_\CH:\Linf{\CF}{d}\to \Linf{\CH}{}$ be the corresponding CRMs.

\begin{Definition}
A CRM $\rho_\CG$ is conditional law-invariant if $\rho_\CG(X)=\rho_\CG(Y)$ whenever the $\CG$-conditional distributions $\mu_X(\cdot|\CG)$ and $\mu_Y(\cdot|\CG)$ of $X,Y\in\Linf{\CF}{d}$ are equal, i.e.\ if $\PW(X\in A\;|\;\CG)=\PW(Y\in A\;|\;\CG)$ for all Borel sets $A\in\CB(\R^d)$. In case $\CG=\{\emptyset,\Omega\}$ is trivial, conditional law-invariance of $\rho_\CG$ is also referred to as law-invariance.
\end{Definition}


In the law-invariant case we will often have to require a little more regularity of the underlying probability space $(\WR,\CF,\PW)$:

\begin{Definition}
We say that $(\WR,\CF,\PW)$ is
\begin{description}
\item[atomless,] if $(\WR,\CF,\PW)$ supports a random variable with continuous distribution;
\item[conditionally atomless given $\bf \CH\subset\CF$,] if $(\WR,\CF,\PW)$ supports a random variable with continuous distribution which is independent of $\CH$.
\end{description}
\end{Definition}


The next lemma shows that conditional law-invariance is passed from $\rho_\CG$ (forward) to $\rho_\CH$ by  strong consistency. The proof is based on \cite{Follmer2014}.

\begin{Lemma}\th\label{l:veerblawinv}
If $\{\rho_\CG,\rho_\CH\}$ is strongly consistent and $\rho_\CG$ is conditionally law-invariant, then $\rho_\CH$ is also conditionally law-invariant.   
\end{Lemma}

\begin{proof}
Let $X,Y\in\Linf{\CF}{}$ such that $\mu_X(\cdot|\CH)=\mu_Y(\cdot|\CH)$ and let $A:=\{\rho_\CH(X)>\rho_\CH(Y)\}\in\CH$. Then the random variables $X\ind_A$ and $Y\ind_A$ have the same conditional distribution given $\CG$. 
As $\rho_\CG$ is conditionally law-invariant and strongly consistent with $\rho_\CH$ we obtain
\begin{align*}
\rho_\CG\Big(f_{\rho_\CH}^{-1}\big(\rho_\CH(X)\ind_A+\rho_\CH(\veczero)\ind_{A^C}\big)\vecone\Big)&=\rho_\CG(X\ind_A)=\rho_\CG(Y\ind_A)\\
&=\rho_\CG\Big(f_{\rho_\CH}^{-1}\big(\rho_\CH(Y)\ind_A+\rho_\CH(\veczero)\ind_{A^C}\big)\vecone\Big).
\end{align*}
On the other hand, by strict antitonicity of $\rho_\CG$ and $f^{-1}_{\rho_\CH}$
$$\rho_\CG\Big(f_{\rho_\CH}^{-1}\big(\rho_\CH(X)\ind_A+\rho_\CH(\veczero)\ind_{A^C}\big)\vecone\Big)\geq\rho_\CG\Big(f_{\rho_\CH}^{-1}\big(\rho_\CH(Y)\ind_A+\rho_\CH(\veczero)\ind_{A^C}\big)\vecone\Big),$$
and the inequality is strict with positive probability if $\PW(A)>0$. Thus $A$ must be a $\PW$-nullset and interchanging $X$ and $Y$ in the definition of $A$ shows that indeed  $\rho_\CH(X)=\rho_\CH(Y)$.
\end{proof}


While in \thref{l:riskanti2} we had to require that the strongly consistent pair $\{\rho_\CG,\rho_\CH\}$ satisfies $\CH=\CF$, in this section we in some sense require the opposite extreme, namely that $\CG=\{\emptyset,\WR\}$ is trivial while $\CH\subseteq\CF$.

\begin{Assumption}
For the rest of the section we assume that $\CG=\{\emptyset,\WR\}$. For simplicity we will write $\rho:=\rho_\CG=\rho_{\{\emptyset,\WR\}}$.
\end{Assumption}

\begin{Lemma}\th\label{l:4neu}
Let $\{\rho,\rho_\CH\}$ be strongly consistent and suppose that $\rho$ is law-invariant (and thus $\rho_\CH$ is conditionally law-invariant by \thref{l:veerblawinv}).
If $(\WR,\CH,\PW)$ is an atomless probability space and $X\in\Linf{\CF}{d}$ is independent of $\CH$, then 
$$f_{\rho_\CH}^{-1}\big(\rho_\CH(X)\big)=f_{\rho}^{-1}\big(\rho(X)\big).$$ 
\end{Lemma}
The proof of \thref{l:4neu} is adapted from \cite{Kupper2009}.
\begin{proof}
We distinguish three cases:
\begin{itemize}
\item Suppose that $f_{\rho_{\CH}}^{-1}\big(\rho_{\CH}(X)\big)\leq f_{\rho}^{-1}\big(\rho(X)\big)$ and strictly smaller with positive probability. Then by strong consistency
\begin{align*}
f_{\rho}^{-1}\big(\rho(X)\big)&=f_{\rho}^{-1}\left(\rho\left(f_{\rho_{\CH}}^{-1}\big(\rho_{\CH}(X)\big)\vecone\right)\right)\\
&<f_{\rho}^{-1}\left(\rho\left(f_{\rho}^{-1}\big(\rho(X)\big)\vecone\right)\right)=f_{\rho}^{-1}\big(\rho(X)\big),
\end{align*} by strict antitonicity of $\rho$ which is a contradiction.
\item Analogously it follows that it is not possible that $f_{\rho_{\CH}}^{-1}\big(\rho_{\CH}(X)\big)\geq f_{\rho}^{-1}\big(\rho(X)\big)$ and $\PW(f_{\rho_{\CH}}^{-1}\big(\rho_{\CH}(X)\big)> f_{\rho}^{-1}\big(\rho(X)\big))>0$. 

\item There exist $A,B\in\CH$ such that $\PW(A)=\PW(B)>0$ and  
$$f_{\rho_{\CH}}^{-1}\big(\rho_{\CH}(X)\big)>f_{\rho}^{-1}\big(\rho(X)\big)\text{ on }A\text{ and }f_{\rho_{\CH}}^{-1}\big(\rho_{\CH}(X)\big)<f_{\rho}^{-1}\big(\rho(X)\big)\text{ on }B.$$ 
Then we have for an arbitrary $m=a\vecone$ where $a\in\R$ that
\begin{align}
\rho(X\ind_A+m\ind_{A^C})&=\rho\left(f_{\rho_{\CH}}^{-1}\big(\rho_{\CH}(X\ind_A+m\ind_{A^C})\big)\vecone\right)\notag\\
&=\rho\left(f_{\rho_{\CH}}^{-1}\big(\rho_{\CH}(X)\big)\ind_A\vecone+m\ind_{A^C}\right)\notag\\
&<\rho\left(f_{\rho}^{-1}\big(\rho(X)\big)\ind_A\vecone+m\ind_{A^C}\right)\label{eq:431}
\end{align}
and similarly
\begin{equation}\label{eq:432}\rho(X\ind_B+m\ind_{B^C})>\rho\left(f_{\rho}^{-1}\big(\rho(X)\big)\ind_B\vecone+m\ind_{B^C}\right).
\end{equation}
However, as $X$ is independent of $\CH$ the random vector
$X\ind_A+m\ind_{A^C}$ has the same distribution under $\PW$ as $X\ind_B+m\ind_{B^C}$.
Note that also 
$f_{\rho}^{-1}\big(\rho(X)\big)\ind_A+a\ind_{A^C}$ and $ f_{\rho}^{-1}\big(\rho(X)\big)\ind_B+a\ind_{B^C}$ share the same distribution under $\PW$. 
Hence, as $\rho$ is law-invariant, \eqref{eq:431} and \eqref{eq:432} yield a contradiction.
\end{itemize}
\end{proof}

Now we are able to extend the representation result of \cite{Follmer2014} to multivariate CRMs.

\begin{Theorem}\th\label{t:1} Let $(\WR,\CH,\PW)$ be atomless and let  $(\WR,\CF,\PW)$ be  conditionally atomless given $\CH$. Suppose that $\rho$ is law-invariant. Then, $\{\rho, \rho_\CH\}$ is strongly consistent if and only if $\rho$ and $\rho_\CH$ are of the form
\begin{equation}\label{eq:ace2}
\rho(X)=g\left(f_{u}^{-1}\big(\EW{\PW}{u(X)}\big)\right)\quad \mbox{for all $X\in\Linf{\CF}{d}$}
\end{equation}
and
\begin{equation}\label{eq:ace}
\rho_\CH(X)=g_{\CH}\left(f_{u}^{-1}\big(\BEW{\PW}{u(X)}{\CH}\big)\right) \quad \mbox{for all $X\in\Linf{\CF}{d}$}
\end{equation}
where $u:\R^d\to\R$ is strictly increasing and continuous, $f_{u}^{-1}:\Imag f_{u}\to\R$ is the inverse function of $$f_{u}:\R\to\R;\quad x\mapsto u(x\vecone)$$ and $g:\R\to\R$ and $g_\CH:\Linf{\CH}{}\to\Linf{\CH}{}$ are strictly antitone, fulfill the Lebesgue property, $0\in\Imag g\cap\Imag g_\CH$, and $g_\CH$ is $\CH$-local.

In particular, for any CRM of type \eqref{eq:ace2} (or \eqref{eq:ace}) we have that $g=f_\rho$ ($g_{\CH}=f_{\rho_\CH}$), where $f_\rho$ and $f_{\rho_\CH}$ are defined in \thref{def:f}.
\end{Theorem}

The common function $u:\R^d\to\R$ appearing in \eqref{eq:ace2} and \eqref{eq:ace} can be seen as a multivariate utility where  $u$ being strictly increasing means that $x,y\in \R^d$ with $x\geq y$ and $x\neq y$ implies $u(x)>u(y)$. So $f_{u}^{-1}\big(\EW{\PW}{u(\cdot)}\big)$ and $f_{u}^{-1}\big(\BEW{\PW}{u(\cdot)}{\CH}\big)$ are (conditional) certainty equivalents -- in the univariate case ($d=1$) we clearly have $f_u^{-1}=u^{-1}$. Thus if $\rho$ and/or $\rho_\CH$ in \thref{t:1} are normalized on constants (and hence $f_\rho\equiv-\operatorname{id}$ or $f_{\rho_\CH}\equiv-\operatorname{id}$), then $\rho$ and/or $\rho_\CH$ equal (minus) certainty equivalents. But \eqref{eq:ace2} and \eqref{eq:ace} also comprise other prominent classes of risk measures. For instance if $f_\rho=-f_u$ or $f_{\rho_\CH}=-f_u$, then
$\rho_\CH(X)=-\EW{\PW}{u(X)}$ is an multivariate expected utility whereas $\rho_\CH(X)=-\BEW{\PW}{u(X)}{\CH}$ is a multivariate conditional expected utility.

\begin{proof} For the last assertion of the theorem note that since $u$ is a deterministic function, we have for $\alpha\in\Linf{\CH}{}$ that
\begin{align*}
f_{\rho_\CH}(\alpha)&=\rho_\CH(\alpha\vecone)=g_{\CH}\left(f_{u}^{-1}\big(\BEW{\PW}{u(\alpha\vecone)}{\CH}\big)\right)\\
&=g_{\CH}\left(f_{u}^{-1}\big(f_{u}(\alpha)\big)\right)=g_{\CH}(\alpha)
\end{align*}
and analogously we obtain $f_\rho\equiv g$.

Next we prove sufficiency in the first statement of the theorem: Let $\rho_\CH$ and $\rho$ be as in \eqref{eq:ace} and \eqref{eq:ace2}.
It is easily verified that $\rho_\CH$ and $\rho$ are (conditionally) law-invariant CRMs. Furthermore, since $f_{u}^{-1}$ is strictly increasing and $g_{\CH}$ is strictly antitone and $\CH$-local, we have for each $X,Y\in\Linf{\CF}{d}$ with $\rho_\CH(X)\geq \rho_\CH(Y)$ that 
$$\BEW{\PW}{u(X)}{\CH}\leq \BEW{\PW}{u(Y)}{\CH}.$$
But this implies that also $\EW{\PW}{u(X)}\leq \EW{\PW}{u(Y)}$ and thus that $\rho(X)\geq\rho(Y)$, i.e.\ $\{\rho,\rho_\CH\}$ is strongly consistent. 

Now we prove necessity in the first statement of the theorem: We assume in the following that $\rho$ and $\rho_\CH$ are normalized on constants and follow the approach of \cite{Follmer2014}~Theorem 3.4. The idea is to introduce a preference order $\prec$ on multivariate distributions $\mu,\nu$ on $(\R^d,{\cal B}(\R^d))$ with bounded support given by
$$\mu\prec\nu\quad\Longleftrightarrow\quad \rho(X)>\rho(Y),\quad \mbox{with $X\sim\mu$ and $Y\sim\nu$}.$$ Here ${\cal B}(\R^d)$ denotes the Borel-$\sigma$-algebra on $\R^d$ and $X\sim\mu$ means that the distribution of $X\in \Linf{\CF}{d}$ under $\PW$ is $\mu$.
It is well-known that if this preference order fulfills a set of conditions, then there exists a von Neumann-Morgenstern representation, that is 
\begin{equation}\label{eq:vonNeumann}
\mu\prec\nu\quad\Longleftrightarrow\quad \int u(x)\,\mu(dx)<\int u(x)\,\nu(dx),
\end{equation}
where $u:\R^d\to\R$ is a continuous function. Sufficient conditions to guarantee \eqref{eq:vonNeumann} are that $\prec$ is continuous and fulfills the independence axiom; cf.\  \cite{Follmer2011}~Corollary~2.28. We refer to \cite{Follmer2011} for a definition and comprehensive discussion of preference orders and the mentioned properties. 
Suppose for the moment that we have already proved \eqref{eq:vonNeumann}. Note that strict antitonicity of $\rho$ implies that $\delta_x\succ\delta_y$ whenever $x,y\in \R^d$ satisfy $x\geq y$ and $x\neq y$. Hence $u(x)=\int u(s)\,\delta_x(ds)>\int u(s)\,\delta_y(ds)=u(y)$, and we conclude that $u$ is necessarily strictly increasing as claimed.  

Now we prove \eqref{eq:vonNeumann}: The proof of continuity of $\prec$ is completely analogous to the corresponding proof in \cite{Follmer2014}~Theorem~3.4, so we omit it here. 
The crucial property is the independence axiom, which states that for any three distributions $\mu,\nu,\vartheta$ such that $\mu\preceq\nu$ and for all $\lambda\in(0,1]$, we have 
$$\lambda\mu+(1-\lambda)\vartheta\preceq\lambda\nu+(1-\lambda)\vartheta.$$
Since $(\WR,\CF,\PW)$ is conditionally atomless given $\CH$, we can find $X,Y,Z\in \Linf{\CF}{d}$ which are independent of $\CH$ such that $X\sim\mu,Y\sim\nu$ and $Z\sim\vartheta$. Furthermore, since $(\WR,\CH,\PW)$ is atomless, we can find an $A\in\CH$ with $\PW(A)=\lambda$. It can be easily seen that $X\ind_A+Z\ind_{A^C}\sim \lambda\mu+(1-\lambda)\vartheta$ and $Y\ind_A+Z\ind_{A^C}\sim \lambda\nu+(1-\lambda)\vartheta$. Moreover, since $\mu\preceq\nu$, we have that $\rho(X)\geq\rho(Y)$. As $\{\rho,\rho_\CH\}$ is strongly consistent and as $\rho$ is  law-invariant, we know from \thref{l:veerblawinv} that $\rho_\CH$ is conditionally law-invariant. This ensures that we can apply \thref{l:4neu} to the random vectors $X$ and $Y$ which are independent of $\CH$. Therefore, by $\CH$-locality of $\rho_\CH$ and recalling \thref{r:34}
\begin{align*}
\rho\left(X\ind_A+Z\ind_{A^C}\right)&=\rho\left(-\rho_\CH\left(X\ind_A+Z\ind_{A^C}\right)\vecone\right) \\ &= \rho\left(-\rho_\CH(X)\ind_A\vecone-\rho_\CH(Z)\ind_{A^C}\vecone\right)\\
																				&=\rho\left(-\rho(X)\ind_A\vecone-\rho_\CH(Z)\ind_{A^C}\vecone\right)\\
																				&\geq\rho\left(-\rho(Y)\ind_A\vecone-\rho_\CH(Z)\ind_{A^C}\vecone\right)\;=\;\rho\left(Y\ind_A+Z\ind_{A^C}\right),
\end{align*}
which is equivalent to $\lambda\mu+(1-\lambda)\vartheta\preceq\lambda\nu+(1-\lambda)\vartheta$. Thus there exists a von Neumann-Morgenstern representation \eqref{eq:vonNeumann} with a continuous and strictly increasing utility function $u:\R^d\to \R$.\\
In the next step we define $f_{u}:\R\to\R;x\mapsto u(x\vecone)$. Then $f_{u}$ is strictly increasing and continuous and thus $f_{u}^{-1}$ exists. Let $\mu$ be an arbitrary distribution on $(\R^d,{\cal B}(\R^d))$ with bounded support and $X\sim\mu$. Then
$$\rho\big(\|X\|_{d,\infty}\vecone\big)\leq \rho(X)\leq \rho\big(-\|X\|_{d,\infty}\vecone\big)$$
and hence 
\begin{align*}
f_{u}(-\|X\|_{d,\infty})&=\int u(x)\;\delta_{-\|X\|_{d,\infty}\vecone}(dx) \leq\int u(x)\;\mu(dx)\\
	&\leq\int u(x)\;\delta_{\|X\|_{d,\infty}\vecone}(dx)= f_{u}(\|X\|_{d,\infty}).
\end{align*} 
The intermediate value theorem now implies the existence of a constant $c(\mu)\in\R$ such that
$$ f_{u}\big(c(\mu)\big)=\int u(x)\;\mu(dx)\quad\Longleftrightarrow\quad c(\mu)= f_{u}^{-1}\left(\int u(x)\;\mu(dx)\right).$$
Finally, since $\delta_{c(\mu)\vecone}\approx \mu$, we have 
\begin{align*}
\rho(X)&=\rho\big(c(\mu)\vecone\big)=-c(\mu)=- f_{u}^{-1}\left(\int u(x)\;\mu(dx)\right)=- f_{u}^{-1}\big(\EW{\PW}{u(X)}\big).
\end{align*}
Hence, we have proved \eqref{eq:ace2} (with $g\equiv -\operatorname{id}$). Define  $$\psi_{\CH}(X):=- f_{u}^{-1}\big(\BEW{\PW}{u(X)}{\CH}\big), \quad X\in \Linf{\CF}{d},$$ then we have seen in the first part of the proof that $\psi_{\CH}$ is a CRM which is strongly consistent with $\rho$. Moreover, $\psi_\CH$ is normalized on constants. Thus it follows by \thref{l:2} that $\rho_\CH=\psi_{\CH}$. If $\rho$ and/or $\rho_\CH$ are not normalized on constants, then considering the normalized CRMs $-f_{\rho}^{-1}\circ\rho$ and $-f_{\rho_\CH}^{-1}\circ\rho_\CH$ as introduced after \thref{r:3}, the result follows from $\rho=f_{\rho}\circ\big(-(-f_{\rho}^{-1}\circ\rho)\big)$ and $\rho_\CH=f_{\rho_\CH}\circ\big(-(-f_{\rho_\CH}^{-1}\circ\rho_\CH)\big)$, i.e.\ $g=f_\rho$ and $g_\CH=f_{\rho_\CH}$.\\
\end{proof}


Recall \thref{l:riskanti2} where we proved that if a multivariate CRM $\rho_\CH$ is strongly consistent in a forward looking way with an aggregation $\rho_\CF$ under full information $\CF$ (and $\rho_\CF$ fulfills \eqref{eq:const}), then the multivariate CRM can be decomposed as in \eqref{eq:help5}. The following \thref{c:decomp} shows that we also obtain such a decomposition \eqref{eq:help5} under law-invariance by requiring strong consistency of $\rho_\CH$ in a backward looking way with $\rho$ given trivial information $\{\emptyset,\Omega\}$.

When stating \thref{c:decomp} we will need an extension of $f_{\rho_\CH}$ to $\Linf{\CF}{}$: Suppose that the process $\R\ni a\mapsto f_{\rho_\CH}(a)$ allows for a continuous realization. Due to the fact that $\rho_\CH$ is strictly antitone and $\CH$-local, we can find a possibly different realization $f_{\rho_\CH}(\cdot,\cdot)$ such that $\widetilde{f}_{\rho_\CH}:\R\times\WR\to\R: x\mapsto f_{\rho_\CH}(x,\om)$ is continuous and strictly decreasing in the first argument for all $\om\in\WR$. Note that there exists a well-defined inverse $\widetilde{f}_{\rho_\CH}^{-1}(\cdot,\om)$ of $\widetilde{f}_{\rho_\CH}(\cdot,\om)$ for all $\om\in\WR$. Now define the functions 
\begin{equation}\label{eq:help6}\bar{f}_{\rho_\CH}:\Linf{\CF}{}\to\Linf{\CF}{};\quad F\mapsto \widetilde{f}_{\rho_\CH}(F(\om),\om)\end{equation}
and $$\bar{f}_{\rho_\CH}^{-1}:\Imag\bar{f}_{\rho_\CH}\to \Linf{\CF}{}; \quad F\mapsto \widetilde{f}_{\rho_\CH}^{-1}(F(\om),\om ),$$  
where we with the standard abuse of notation identify the random variable $\widetilde{f}_{\rho_\CH}(F(\om),\om)$ or $\widetilde{f}_{\rho_\CH}^{-1}(F(\om),\om )$ with the equivalence classes they generate in $\Linf{\CF}{}$.

\noindent By construction of $\bar f_{\rho_\CH}$ we have that
$$\bar f_{\rho_\CH}(\Linf{\CJ}{})\subseteq\Linf{\CJ}{}$$ for all $\sigma$-algebras $\CJ$ such that  $\sigma\left(f_{\rho_\CH}(a,\cdot),a\in\R\right)\subseteq\CJ\subseteq\CF$, c.f.\ \cite{Hoffmann2016}~Lemma~3.1. By definition $\bar f_{\rho_\CH}$ is also $\CF$-local and has the Lebesgue property due to continuity of $\R\ni a\mapsto \widetilde f_{\rho_\CH}(a,\omega)$. Moreover, $\CH$-locality and continuity also imply that indeed $\bar f_{\rho_\CH}(X)=f_{\rho_\CH}(X)$ for all $X\in\CH$ (approximation by simple random variables), so $\bar f_{\rho_\CH}$ is indeed an extension of $f_{\rho_\CH}$ to  $\Linf{\CF}{}$.

\begin{Theorem}\th\label{c:decomp}
Under the same conditions as in \thref{t:1} let $\{\rho, \rho_\CH\}$ be strongly consistent.  Then $\rho$ can be decomposed as
$$\rho=\eta\circ \Lambda,$$
where $$\Lambda:\Linf{\CF}{d}\to\Linf{\CF}{};\quad  X\mapsto -f_\rho\left(f^{-1}_u\left(u(X)\right)\right)$$ is a $\{\emptyset,\WR\}$-conditional aggregation function,
$$\eta:\Imag \Lambda\to\R;\quad F\mapsto -U^{-1}\left(\EW{\PW}{U(F)}\right)$$ is a law-invariant univariate certainty equivalent given by the (deterministic) utility
$$U:\Imag\rho\to\R;\quad  a\mapsto f_u\left(f_\rho^{-1}(-a)\right)$$ 
which is strictly increasing and continuous. Here $u:\R^d\to \R$ is the multivariate utility function from \thref{t:1}.  
\\If the function $\R\ni a\mapsto f_{\rho_\CH}(a)$ has a continuous realization, then $\rho_\CH$ can be decomposed as 
$$\rho_\CH=\eta_\CH\circ\Lambda_\CH,$$
with
$$\eta_\CH\left(\Lambda_\CH(X)\right)=-\Lambda_\CH(X),\quad\text{for all }X\in\Linf{\CH}{d},$$
where
\begin{itemize}
\item $\Lambda_\CH:\Linf{\CF}{d}\to\Linf{\CF}{};\quad X\mapsto -\bar f_{\rho_\CH}\left(f_u^{-1}\left(u(X)\right)\right)$ is a $\sigma\left(f_{\rho_\CH}(a,\cdot):a\in\R\right)$-conditional aggregation function ($f_{\rho_\CH}(a,\cdot)$ denotes a continuous realization with strictly increasing paths);
\item $\eta_\CH:\Imag \Lambda_\CH\to\Linf{\CH}{};\quad  F\mapsto -U_\CH^{-1}\left(\BEW{\PW}{U_\CH(F)}{\CH}\right)$ is a univariate conditional certainty equivalent;
\item the stochastic utility $U_\CH:\Imag\Lambda_\CH\to\Linf{\CF}{};\quad F\mapsto f_u\left(\bar f_{\rho_\CH}^{-1}(-F)\right)$ is strictly isotone, $\CF$-local, fulfills the Lebesgue property and $U_\CH^{-1}(\Imag U_\CH\cap\Linf{\CH}{})\subseteq\Linf{\CH}{}$;
\item $\bar f_{\rho_\CH}$ is given in \eqref{eq:help6}.
\end{itemize}
Moreover, it holds that 
\begin{equation}\label{eq:help7}U_\CH\circ \Lambda_\CH=u=U\circ\Lambda\end{equation} 
are deterministic and independent of the chosen information $\CH$ or $\{\WR,\emptyset\}$.\\
Finally we also have that $f_{\Lambda_\CH}^{-1}\circ \Lambda_\CH=f_u^{-1}\circ u=f_{\Lambda}^{-1}\circ\Lambda$, i.e.\ $\{\Lambda,\Lambda_\CH\}$ is strongly consistent as defined in \eqref{eq:consoflam}.
\end{Theorem}

\begin{proof}
By \thref{t:1} we have that
\begin{align*}\rho_\CH(X) &=f_{\rho_\CH}\left(f_{u}^{-1}\left(\BEW{\PW}{u(X)}{\CH}\right)\right)\\
&=\bar f_{\rho_\CH}\left(f_{u}^{-1}\left(\BEW{\PW}{f_{u}\left(\bar f_{\rho_\CH}^{-1}\left(\bar f_{\rho_\CH}\left(f_{u}^{-1}\left(u(X)\right)\right)\right)\right)}{\CH}\right)\right),
\end{align*}
where $u$ and $f_u$ are given in \thref{t:1}. Hence, recalling the definitions of $U_\CH$, $\eta_\CH$, and $\Lambda_\CH$, we have $\rho_\CH=\eta_\CH\circ\Lambda_{\CH}$.
It can be readily seen that $U_\CH$ as well as $U_\CH^{-1}$, and thus also $\Lambda_\CH$, are $\CF$-local, strictly isotone, and fulfill the Lebesgue property. As  $\bar f_{\rho_\CH}(\Linf{\CJ}{})\subseteq\Linf{\CJ}{}$ for all $\sigma$-algebras $\CJ$ such that  $\sigma\left(f_{\rho_\CH}(a,\om):a\in\R\right)\subseteq\CJ\subseteq\CF$, the same also applies to $\Lambda_\CH=-\bar f_{\rho_\CH}\circ f_u^{-1}\circ u$ and we conclude that $\Lambda_\CH$ is a $\sigma\left(f_{\rho_\CH}(a,\om):a\in\R\right)$-conditional aggregation function.
Moreover, for $X\in\Linf{\CH}{d}$
$$\eta_\CH\left(\Lambda_\CH(X)\right)=\bar f_{\rho_\CH}\left(f_u^{-1}(u(X))\right)=- U_\CH^{-1}\big(u(X)\big)=-\Lambda_{\CH}(X).$$
The result for $\rho$ follows similarly to the proof above without requiring a continuous realization and by using the canonical extension of $f_\rho$ from $\R$ to $\Linf{\CF}{d}$, i.e.\ $\bar f_\rho (F)(\om)=f_\rho(F(\om))$ for all $\om\in\WR$ and $F\in\Linf{\CF}{}$. 
\end{proof}

We remark that \eqref{eq:help7} is the crucial fact which ensures that $\rho$ and $\rho_\CH$ are strongly consistent and (conditionally) law-invariant.

In \th\ref{c:decomp} we have seen that basically every CRM which is strongly consistent with a law-invariant CRM under trivial information can be decomposed into a conditional aggregation function and a univariate conditional certainty equivalent. For the rest of this section we study the effect of additional properties of the CRMs on this decomposition. For instance, we want to identify conditions under which the univariate conditional certainty equivalent is generated by a deterministic (instead of a stochastic) utility function; see \thref{c:new1}. 
Also we study what happens if the univariate CRMs $\eta$ and $\eta_\CH$ from \th\ref{c:decomp} are required to be strongly consistent; see \thref{c:etastrongcons}. 


\begin{Corollary}\th\label{c:new1}
In the situation of \thref{c:decomp}, if $\rho$ is normalized on constants, then  
$$\Lambda(X)=f^{-1}_u(u(X)),\quad X\in \Linf{\CF}{d}, $$ and $$\eta(F)=\rho(F\vecone) =-f^{-1}_u(\EW{\PW}{f_u(F)}),\quad F\in \Linf{\CF}{}.$$ If $\rho_\CH$ is normalized on constants, then similarly  $$\Lambda_\CH(X)=f^{-1}_u(u(X)),\quad X\in \Linf{\CF}{d}, $$ and 
$$\eta_\CH(F)=\rho_\CH(F\vecone) =-f^{-1}_u(\BEW{\PW}{f_u(F)}{\CH}),\quad F\in \Linf{\CF}{}.$$ In particular the univariate conditional certainty equivalent $\eta_\CH$ is now given by the deterministic univariate utility function $f_u$, and thus $\eta_\CH$ is conditionally law-invariant. \\
If both $\rho$ and $\rho_\CH$ are normalized on constants, then $\Lambda=\Lambda_\CH$.   
\end{Corollary}

\begin{Remark}\th\label{r:cashconvex}
Suppose that $\rho$ and $\rho_\CH$ from \thref{c:decomp} are normalized on constants and   that for all $F,G\in\Linf{\CF}{}$, $m,\lambda\in\R$ with $\lambda\in(0,1)$
\begin{equation}\label{eq:cashadd}
\rho(F\vecone+m\vecone)=\rho(F\vecone)-m
\end{equation}
as well as
\begin{equation}\label{eq:convexity}
\rho\big(\lambda F\vecone+(1-\lambda)G\vecone\big)\leq \lambda\rho(F\vecone)+(1-\lambda)\rho(G\vecone).
\end{equation}
Recalling \thref{c:new1} it follows that $\eta(F)=\rho(F\vecone)$ is cash-additive \eqref{eq:cashadd} and convex \eqref{eq:convexity}. Since $f_u$ is a deterministic function it can be easily checked that $\eta$ and $\eta_\CH$ are strongly consistent (conditionally) law-invariant univariate CRMs. Therefore we are in the framework of \cite{Follmer2014}. There it is shown that the univariate CRMs must be either linear or of entropic type, i.e. 
$$f_u(x)=ax+b\quad\text{or}\quad f_u(x)=-a e^{-\beta x}+b,\quad x\in \R, $$ for constants $a,b,\beta\in\R$ with $a,\beta>0$, which implies that
$$\eta_\CH(F)=\BEW{\PW}{-F}{\CH}\quad\text{or}\quad \eta_\CH(F)=\frac1\beta\log\left(\BEW{\PW}{e^{-\beta F}}{\CH}\right)$$
and similarly for $\eta$. Clearly, this also has consequences for the aggregation function $\Lambda=\Lambda_\CH=f_u^{-1}\circ u$ since $x\mapsto u(x\vecone)=f_u(x)$ is either of linear or exponential form. For instance, a possible aggregation would be given by $u(x_1,\ldots, x_d)=a\sum_{i=1}^d w_i x_i+b,$
where $w_i\in (0,1)$ for $i=1,...,d$ such that $\sum_{i=1}^dw_i=1$, because $f_u(x)=ax+b$. In this case the aggregation function is simply 
$\Lambda(x)=\sum_{i=1}^dw_ix_i$.
\end{Remark}

\begin{Corollary}\th\label{c:etastrongcons}
In the situation of \thref{c:decomp}, suppose that $\eta$ and $\eta_\CH$ are defined on all of  $\Linf{\CF}{}$. Then $\{\eta,\eta_\CH\}$ are strongly consistent if and only if $$\eta=-\Wu^{-1}\left(\EW{\PW}{\Wu(F)}\right)\quad\text{and}\quad\eta_\CH=-\Wu^{-1}\left(\BEW{\PW}{\Wu(F)}{\CH}\right)$$
for a continuous and strictly increasing utility function $\Wu:\R\to\R$. 
Moreover, the corresponding (conditional) aggregation functions are given by
$$\Lambda=-f_\rho\circ f^{-1}_u\circ u\quad\text{and}\quad\Lambda_\CH=-f_\rho\circ f_u^{-1}\circ a_\CH\circ u,$$
where $a_\CH(F)=\alpha F+\beta$, $F\in \Linf{\CF}{}$, is a positive affine transformation given by $\alpha,\beta\in \Linf{\CH}{}$ with $\PW(\alpha>0)=1$.
\end{Corollary}

\begin{proof}
As $\eta$ is law-invariant, it follows from \thref{l:veerblawinv} that $\eta_\CH$ is conditionally law-invariant. 
Moreover, $f_\eta\equiv f_{\eta_\CH}\equiv -\id$, i.e.\ $\eta$ and $\eta_\CH$ are normalized on constants. Thus by \thref{t:1} we obtain that 
$$\eta=-\Wu^{-1}\left(\EW{\PW}{\Wu(F)}\right)\quad\text{and}\quad\eta_\CH=-\Wu^{-1}\left(\BEW{\PW}{\Wu(F)}{\CH}\right)$$
for a continuous and strictly increasing function $\Wu:\R\to \R$. 
It follows from \th\ref{l:43} below that $U$ as well as $U_\CH$ are affine transformations of $\Wu$. This in turn implies that $U_\CH=\tilde a_\CH\circ U$,
where $\tilde a_\CH(F)=\tilde\alpha F+\tilde\beta$ for $\tilde\alpha,\tilde\beta\in\Linf{\CH}{}$ with $\PW(\tilde\alpha>0)=1$. Finally we obtain that the $\sigma\left(f_{\rho_\CH}(a,\om),a\in\R\right)$-conditional aggregation function $\Lambda_\CH$ is given by
$$\Lambda_\CH=U_\CH^{-1}\circ u=U^{-1}\circ \tilde a_\CH^{-1}\circ u=-f_\rho\circ f_u^{-1}\circ \tilde a_\CH^{-1}\circ u.$$
Since the inverse $a_\CH:=\tilde a_\CH^{-1}$ of an affine function is affine the result follows.
\end{proof}

\begin{Lemma}\th\label{l:43}
Let $U_\CH$ be the stochastic utility from \th\ref{c:decomp} and let $\WU_\CH:\Imag \Lambda_\CH\to\Linf{\CF}{}$ be another function which is strictly isotone, $\CF$-local, fulfills the Lebesgue property and $\WU_\CH(\Imag \Lambda_\CH\cap\Linf{\CH}{})\subseteq\Linf{\CH}{}$, such that
\begin{equation}\label{eq:1}\WU^{-1}_\CH\left(\BEW{\PW}{\WU_\CH(F)}{\CH}\right)=U^{-1}_\CH\left(\BEW{\PW}{U_\CH(F)}{\CH}\right),\quad\text{for all }F\in\Imag\Lambda_\CH.\end{equation}
Then $\WU_\CH$ is an $\CH$-measurable positive affine transformation of $U_\CH$, i.e.\ there exist $\alpha,\beta\in\Linf{\CH}{}$ with $\PW(\alpha>0)=1$ such that $\WU_\CH(F)=\alpha U_\CH(F)+\beta$ for all $F\in\Imag\Lambda_\CH$. 
\end{Lemma}

\begin{proof}
We have seen in \th\ref{c:decomp} that $U_\CH\circ\Lambda_\CH =u$, where $u$ is strictly increasing and continuous. Thus 
$$\CX:=\Imag U_\CH=u(\Linf{\CF}{d})\subseteq\Linf{\CF}{}$$
and it follows that for all $F\in\CX$ there exists a sequence of $\CF$-simple random variables $(F_n)_{n\in\N}\subseteq\CX$ such that $F_n\to F$ $\PW$-a.s. 
Moreover, by the intermediate value theorem we can find for each $X,Y\in\Linf{\CF}{d}$ and $\lambda\in\Linf{\CF}{}$ with $0\leq\lambda\leq 1$ a random variable $Z$ such that $\min\{-\|X\|_{d,\infty},-\|Y\|_{d,\infty}\}\leq Z\leq \max\{\|X\|_{d,\infty},\|Y\|_{d,\infty}\}$ and for all $\PW$-almost all $\om\in\WR$ 
$$\lambda(\om) u\big(X(\om)\big)+(1-\lambda) u\big(Y(\om)\big)=u\big(Z(\om)\vecone\big)$$
where $X(\cdot),Y(\cdot)$ and $\lambda(\cdot)$ are arbitrary representatives of $X,Y$ and $\lambda$.
Indeed, it can be shown by a measurable selection argument that $Z$ can be chosen to be $\CF$-measurable and hence $\CX$ is $\CF$-conditionally convex in the sense that $\lambda F+(1-\lambda)G\in \CX$ for all $F,G\in \CX$ and $\lambda\in \Linf{\CF}{}$ with $0\leq \lambda\leq 1$.

Next define the strictly isotone and $\CF$-local function 
$$ V_\CH:\CX\to\Linf{\CF}{};\;X\mapsto \WU_\CH\left(U_\CH^{-1}(F)\right),$$
that is $\WU_\CH=V_\CH\circ U_\CH$. Moreover, it easily follows that $V_\CH$ fulfills the Lebesgue property and $V_\CH(\CX\cap\Linf{\CH}{})\subseteq\Linf{\CH}{}$.
We show that $V_\CH$ is an affine function, that is $V_\CH(F)=\alpha F+\beta$ for all $F\in\CX$, where $\alpha,\beta\in\Linf{\CF}{}$. Note that affinity can be equivalently expressed via $V_\CH(\lambda F+(1-\lambda)G)=\lambda V_\CH(F)+(1-\lambda) V_\CH(G)$ for all $F,G\in\CX$ and $\lambda\in\Linf{\CF}{}$ with $0\leq\lambda\leq1$. \\
We suppose that $V_\CH$ is not affine, i.e.\ there are $F,G\in\CX$ and $\lambda\in\Linf{\CF}{}$ with $0\leq\lambda\leq1$ such that 
\begin{equation}\label{eq:gneq}\PW\left(V_\CH(\lambda F+(1-\lambda)G)\neq\lambda V_\CH(F)+(1-\lambda) V_\CH(G)\right)>0.\end{equation}

First note that it suffices to assume that \eqref{eq:gneq} holds for deterministic $F,G$ and $\lambda$. To see this suppose that $ V_\CH$ is affine on deterministic values, but not on the whole of $\CX$, i.e.\ \eqref{eq:gneq} holds for some $F,G\in\CX$ and $\lambda\in\Linf{\CF}{}$ with $0\leq\lambda\leq1$. We know that there exist sequences of $\CF$-simple functions $(F_n)_{n\in\N},(G_n)_{n\in\N}\subset\CX\cap\CS$ and $(\lambda_n)_{n\in\N}\subset\Linf{\CF}{}\cap\CS$ with $0\leq\lambda_n\leq1$ for all $n\in\N$ such that $F_n\to F,G_n\to G,\lambda_n\to\lambda$ $\PW$-a.s., where $\CS$ was defined in the proof of \th\ref{T1}. Without loss of generality we might assume that $F_n=\sum_{i=1}^{k_n}F_i^n\ind_{A_i^n},G_n=\sum_{i=1}^{k_n}G_i^n\ind_{A_i^n}$ and $\lambda_n=\sum_{i=1}^{k_n}\lambda_i^n\ind_{A_i^n}$ have the same disjoint $\CF$-partition $(A_i^n)_{i=1,...,k_n}$. By the $\CF$-locality and Lebesgue property and since $F_i^n,G_i^n,\lambda_i^n\in\R$ for all $i=1,...,k_n$ and $n\in\N$ we have
\begin{align*}
 V_\CH(\lambda F+(1-\lambda)G)&=\lim_{n\to\infty}  V_\CH(\lambda_n F_n+(1-\lambda_n)G_n)\\
&=\lim_{n\to\infty}  V_\CH\left(\sum_{i=1}^{k_n}(\lambda_i^n F^n_i+(1-\lambda^n_i)G^n_i)\ind_{A_i^n}\right)\\
&=\lim_{n\to\infty} \sum_{i=1}^{k_n} V_\CH\big(\lambda_i^n F^n_i+(1-\lambda^n_i)G^n_i\big)\ind_{A_i^n}\\
&=\lim_{n\to\infty} \sum_{i=1}^{k_n}\Big(\lambda^n_i V_\CH(F^n_i)+(1-\lambda_i^n) V_\CH(G_i^n)\Big)\ind_{A_i^n}\\
&=\lim_{n\to\infty}\lambda_n V_\CH(F_n)+(1-\lambda_n) V_\CH(G_n)\\
&=\lambda V_\CH(F)+(1-\lambda) V_\CH(G),
\end{align*}
which contradicts \eqref{eq:gneq}.
Moreover we assume that $0<\lambda<1$ since otherwise this would also contradict \eqref{eq:gneq}. Finally, we assume w.l.o.g.\ that 
$$A:=\{V_\CH(\lambda F+(1-\lambda)G)<\lambda V_\CH(F)+(1-\lambda) V_\CH(G)\}\in\CH$$
has positive probability. 
Next define $H_1:=F\ind_A+G\ind_{A^C}$ and $H_2:=G$, then $H_i\in\CX\cap\Linf{\CH}{}, i=1,2$ and by $\CF$-locality of $V_\CH$
$$V_\CH(\lambda H_1+(1-\lambda)H_2)\leq \lambda V_\CH(H_1)+(1-\lambda) V_\CH(H_2)$$
and the inequality is strict with positive probability.

Since $(\WR,\PW,\CF)$ is conditionally atomless given $\CH$ there exists a $B\in\CF$ with $\PW(B)=\lambda$ and which is independent of $\CH$. Since $H_1,H_2\in\CX$ and $\CX$ is $\CF$-conditionally convex 
$$H:=H_1\ind_B+H_2\ind_{B^C}\in\CX.$$
Now by $\CF$-locality of $V_\CH$, $V_\CH(\CX\cap\Linf{\CH}{})\subseteq\Linf{\CH}{}$ and $B\perp\!\!\!\!\perp \CH$ we get
\begin{align*}
\BEW{\PW}{ V_\CH\left(H\right)}{\CH}&=\BEW{\PW}{ V_\CH\left(H_1\ind_B+H_2\ind_{B^C}\right)}{\CH}\\
&=V_\CH(H_1)\BEW{\PW}{\ind_B }{\CH}+V_\CH(H_2)\BEW{\PW}{\ind_{B^C} }{\CH}\\
&= V_\CH(H_1)\EW{\PW}{\ind_B }+V_\CH(H_2)\EW{\PW}{\ind_{B^C} }\\
&=\lambda V_\CH(H_1)+(1-\lambda)V_\CH(H_2)\\
&\geq V_\CH(\lambda H_1+(1-\lambda)H_2)\\
&=V_\CH\left(\BEW{\PW}{H_1\ind_B+H_2\ind_{B^C}}{\CH}\right)\\
&=V_\CH\left(\BEW{\PW}{H}{\CH}\right),\\
\end{align*}
and the inequality is strict with positive probability.
Moreover $\CX=\Imag U_\CH$ implies the existence of a $\WH\in\Imag\Lambda_\CH$ such that $H=U_\CH(\WH).$
Finally we get
\begin{align*}
\WU_\CH^{-1}\left(\BEW{\PW}{\WU_\CH(\WH)}{\CH}\right)&=U_\CH^{-1}\left( V_\CH^{-1}\left(\BEW{\PW}{ V_\CH\left(U_\CH(\WH)\right)}{\CH}\right)\right)\\
&=U_\CH^{-1}\left( V_\CH^{-1}\left(\BEW{\PW}{ V_\CH\left(H\right)}{\CH}\right)\right)\\
&\geq U_\CH^{-1}\left( V_\CH^{-1}\left(V_\CH\left(\BEW{\PW}{ H}{\CH}\right)\right)\right)\\
&=U_\CH^{-1}\left(\BEW{\PW}{H}{\CH}\right)\\
&=U_\CH^{-1}\left(\BEW{\PW}{U_\CH(\WH)}{\CH}\right),
\end{align*}
and the inequality is strict with positive probability, since $\WU^{-1}_\CH$ and $U_\CH^{-1}$ are strictly isotone (c.f.\ \th\ref{l:-1}). Thus we have the desired contradiction of \eqref{eq:1} and hence $V_\CH$ is affine, i.e.\ $V_\CH(F)=\alpha F+\beta$ for all $F\in\CX$, where $\alpha,\beta\in\Linf{\CF}{}$. Moreover, since we know that $V_\CH(x)\in\Linf{\CH}{}$ for all $x\in\R\cap\CX$, we obtain that $\alpha,\beta$ are actually $\CH$-measurable. That $\alpha>0$ follows immediately from the fact that $\WU_\CH,U^{-1}_\CH$ are strictly isotone.
\end{proof}

\begin{Remark}\th\label{r:Kromer1}
Our notion of consistency is defined in terms of the multivariate CRMs. In contrast in \cite{Kromer2014} it is a priori assumed that the multivariate CRMs are of the decomposable form $\rho=\eta\circ\Lambda$ as in \eqref{eq:help5}  and they define "consistency" of $\{\rho_\CG,\rho_\CH\}$ by requiring strong consistency of both pairs $\{\eta_\CG,\eta_\CH\}$ and $\{\Lambda_\CG,\Lambda_\CH\}$.
Note that these definitions of consistency are not equivalent, in particular strong consistency of both $\{\eta_\CG,\eta_\CH\}$ and $\{\Lambda_\CG,\Lambda_\CH\}$ does not imply strong consistency of $\{\rho_\CG,\rho_\CH\}$. \cite{Kromer2014} also study the interplay of the strong consistency of $\{\rho_\CG,\rho_\CH\}$ and of strong consistency of both $\{\eta_\CG,\eta_\CH\}$ and $\{\Lambda_\CG,\Lambda_\CH\}$.
As \th\ref{c:etastrongcons} shows in the law-invariant case this requirement is quite restrictive.

\end{Remark}

\section{Consistency of a family of conditional risk measures}
So far we only considered consistency for two multivariate CRMs. In this section we extend our results on strong consistency to families of multivariate CRMs. We begin with some motivating examples.
\begin{Example}[Dynamic risk measures]\th\label{ex:dynamic}
If one is interested in a dynamic risk measurement under growing information in time up to a terminal time $T>0$, this can be modeled by a family of CRMs $(\rho_t)_{t\in[0,T]}$ and a filtration $(\CF_t)_{t\in[0,T]}$ such that $\rho_t:\Linf{\CF_T}{d}\to \Linf{\CF_t}{}$. 
\end{Example}
In systemic risk measurement conditioning on varying information in space rather than in time is of interest. In that situation, as opposed to Example~\ref{ex:dynamic}, the family of multivariate CRMs is not necessarily indexed by a filtration. To exemplify this we recall a multivariate version of the spatial risk measures which have been introduced by \cite{Follmer2014} in a univariate framework.

\begin{Example}[Multivariate spatial risk measures]\th\label{ex:spatial1}
 Let $I=\{1,...,d\}$ denote a set of financial institutions and let $(S,\CS)$ be a measurable space. Each financial institution $i\in I$ can be in some state $s\in S$, and $\WR=S^I=\{\om=(\om)_{i\in I}:\om_i\in S\}$ denotes all possible states of the system. Then the $\sigma$-algebra $\CF_J$ on $\WR$ which is generated by the canonical projections on the $j$-th coordinate for $j\in J$ describes the observable information within the subsystem of financial institutions $J\subseteq I$. Finally let $\PW$ be a probability measure on $(\WR,\CF)$, where $\CF:=\CF_{I}$. Then the risk evolution under varying spatial information can be modeled by the family of CRMs $(\rho_{J})_{J\subseteq I}$, where each $\rho_{J}:\Linf{\CF}{d}\to\Linf{\CF_{J}}{}$, i.e.\ $\rho_{J}$ is the risk of the system given the information on the state of the financial institutions within the subsystem $J$.  
\end{Example}

From the viewpoint of a regulator, systemic risk measurement contingent on information in space is helpful in identifying systemic relevant structures, i.e.~in analyzing questions like: "How much is the system affected given that a specific institution or subgroup of institutions is in distress?", or "How resilient is a specific institution or subgroup of institutions given that the system is in distress?". In Example~\ref{ex:spatial1} the spatial conditioning is based on a $\sigma$-algebra which is generated by all possible states of the institutions within a given subsystem. To treat questions of the type mentioned before one might alternatively consider conditioning with respect to more granular information in space. For instance, in the spirit of the systemic risk measures CoVaR in \cite{Adrian2011} or Systemic Expected Shortfall in \cite{Acharya2010} one could condition on a single crisis event with respect to a given subsystem, e.g.~that all financial institutions within the subsystem are below their individual value-at-risk levels.


\ \

In \thref{ex:dynamic} as well as \thref{ex:spatial1} the families of CRMs are indexed by one-dimensional information structure. However, in \cite{Frittelli2011b}, they propose conditional certainty equivalents based on a one-dimensional information structure caused by the fact that utilities of agents may vary over time: 

\begin{Example}[Conditional certainty equivalents]\th\label{ex:CCE}
Let $(\WR,\CF,(\CF_t)_{t\in\R^+},\PW)$ be an atomless filtered probability space and let $u_t:\R\times\WR\to\R$ be a function which is strictly increasing and continuous in the first argument and $\CF_t$-measurable in the second argument for all $t\in\R^+$. Suppose that the range $\CR_t:=\{u_t(x,\om):x\in\R\}$ is independent of $\om\in\WR$, that $\CR_t\subseteq\CR_s$ for all $s\leq t$, and denote the pathwise inverse function of $u$ by $u_t^{-1}(y)\in\Linf{\CF_t}{}$ for all $y\in\CR_t$, where $u_t(x)$ and $u_t^{-1}(y)$ is the shorthand for  $u_t(x,\cdot)$ and $u_t^{-1}(y,\cdot)$, resp.
Then the backward conditional certainty equivalent is given by
$$C_{s,t}:\Linf{\CF_t}{}\to\Linf{\CF_s}{}; F\mapsto C_{s,t}(F)=-u_s^{-1}\big(\BEW{\PW}{u_t(F)}{\CF_s}\big).$$
It has been shown in \cite{Frittelli2011b} Proposition 1.1 that for a fixed $T\in\R^+$, we have that the family $(C_{t,T})_{t\leq T}$ is consistent, i.e.\ for all $s\leq t\leq T$ 
$$C_{t,T}(F)\geq C_{t,T}(G)\Longrightarrow C_{s,T}(F)\geq C_{s,T}(G)\quad(F,G\in\Linf{\CF_T}{}).$$
\end{Example}
Also in the context of conditioning on spatial information a two-dimensional information structure could be of interest, for example to represent risk measurement policies that differ locally in the financial system.

\begin{Example}[Local regulatory policies]\th\label{ex:LRMP}
In the context of \thref{ex:spatial1}, let $I=\{1,...,d\}$ be a network of financial institutions that is of interest for supervisory authorities associated to different levels with possibly different regulatory policies.  For example, think of $I$ as the European financial system. Then regulatory policies of authorities on the European level might differ from policies on the national levels which again might differ from regional policies. To include these different regulatory viewpoints into the framework of spatial risk measures one could consider a family of CRMs $(\rho_{J,K})_{J\subseteq K\subseteq I}$, where each $\rho_{J,K}:\Linf{\CF_{K}}{d}\to\Linf{\CF_{J}}{}$. Here the first index $J$ has the same meaning as in \thref{ex:spatial1}, i.e.~the risk measurement is performed conditioned on the state of the institutions in subsystem $J$. The second index $K$ identifies the type of regulatory policy on the risk management prevailing in subsystem $K$, for example expected shortfall measures at different significance levels according to European ($K=I$), national, or regional standards. Even though regulatory policies may differ depending on the level of authority, it might still be desirable that these policies  behave consistently in some way, i.e.~the family $(\rho_{J,K})_{J\subseteq K\subseteq I}$ should be consistent not only with respect to the contingent information implied by the index $J$ but also with respect to the different policies implied by the index $K$. In the following, this question will be considered.
\end{Example}

Motivated by the examples above, we will consider the following types of families of CRMs in this section: 
Let $\CI_1$ and $\CI_2$ be sets of sub-$\sigma$-algebras of $\CF$ such that $\CI_1$ contains the trivial $\sigma$-algebra and denote by $\CE:=\{(\CH,\CT)\in\CI_1\times\CI_2:\CH\subseteq\CT\}$. In the following we denote by $\rho_{\CH,\CT}$ a multivariate CRM which maps $\Linf{\CT}{d}$ to $\Linf{\CH}{}$ and we consider families of CRMs of type $(\rho_{\CH,\CT})_{(\CH,\CT)\in \CE}$. 
In order to allow for a comparison of the risks of two random risk factors under different information, we assume for the rest of this section that $\rho_{\CH,\CT_1}(\Linf{\CT_1}{d})=\rho_{\CH,\CT_2}(\Linf{\CT_2}{d})$ for all $(\CH,\CT_1),(\CH,\CT_2)\in\CE$.
Sometimes it will also be convenient to consider only a subfamily of $\CE$ where the second $\sigma$-algebra is fixed. In that case we denote the corresponding index set by $\CE(\CT):=\{\CH\in\CI_1: \CH\subseteq\CT\}$ for $\CT\in\CI_2$. 
Note that the structure of the families of CRMs discussed in \thref{ex:dynamic} and \thref{ex:spatial1} is covered by this framework by letting $\CI_2:=\{\CF\}$.

\begin{Definition}\th\label{d:consforfamily}
A family of CRMs $(\rho_{\CH,\CT})_{(\CH,\CT)\in\CE}$ is strongly consistent if for all $\CG\subseteq\CH\subseteq \CT_1\cap\CT_2$
$$\rho_{\CH,\CT_1}(X)\geq\rho_{\CH,\CT_2}(Y)\Longrightarrow \rho_{\CG,\CT_1}(X)\geq\rho_{\CG,\CT_2}(Y),\quad (X\in\Linf{\CT_1}{},Y\in\Linf{\CT_2}{d}).$$
\end{Definition}
It can be easily checked that the conditional certainty equivalents of \cite{Frittelli2011b} (see \th\ref{ex:CCE}) are strongly consistent. 
Analogously to \thref{l:1} strong consistency is equivalent to the following recursive relation between the CRMs.

\begin{Lemma}
Let $(\rho_{\CH,\CT})_{(\CH,\CT)\in\CE}$ be family of CRMs, then the following statements are equivalent:
\begin{enumerate}
\item $(\rho_{\CH,\CT})_{(\CH,\CT)\in\CE}$ is strongly consistent;
\item For all $\CG\subseteq\CH\subseteq \CT_1\cap\CT_2$ and $X\in\Linf{\CT_1}{d}$ 
$$\rho_{\CG,\CT_1}(X)=\rho_{\CG,\CT_2}\left(f^{-1}_{\rho_{\CH,\CT_2}}\big(\rho_{\CH,\CT_1}(X)\big)\vecone\right).$$
\end{enumerate}
\end{Lemma}

Clearly, our results from the previous sections carry over to families of CRM. We illustrate this in the following by giving the straightforward extensions of \thref{l:riskanti2} and \thref{t:1} to a family of CRMs.
\begin{Theorem}
Let $(\rho_{\CH,\CT})_{(\CH,\CT)\in\CE}$ be a family of strongly consistent CRMs. Moreover, if there exists a $\CT\in\CI_2$ such that
\begin{equation}
f^{-1}_{\rho_{\CT,\CT}}\circ\rho_{\CT,\CT}(x)\in\R,\quad\forall x\in\R^d,
\end{equation}
then each multivariate CRM $\rho_{\CH,\CT}$ of the subfamily $(\rho_{\CH,\CT})_{\CH\in\CE(\CT)}$ which has a continuous realization $\rho_{\CH,\CT}(\cdot,\cdot)$ can be decomposed into a $\CH$-conditional aggregation function $\Lambda_{\CH,\CT}:\Linf{\CT}{d}\to\Linf{\CT}{}$ and a univariate CRM $\eta_{\CH,\CT}:\Imag\Lambda_{\CH,\CT}\to \Linf{\CH}{}$ such that 
$$\rho_{\CH,\CT}=\eta_{\CH,\CT}\circ\Lambda_{\CH,\CT}$$
 and $\rho_{\CH,\CT}(X)=\eta_{\CH,\CT}\big(\Lambda_{\CH,\CT}(X)\big)=-\Lambda_{\CH,\CT}(X) \text{ for all }X\in\Linf{\CH}{d}.$
Moreover, for those $\rho_{\CH,\CT},\CH\in\CE(\CT)$, for which a decomposition exists the corresponding conditional aggregation functions are strongly consistent.
\end{Theorem}

\begin{Theorem}\th\label{t:3}
Let $(\rho_{\CH,\CT})_{(\CH,\CT)\in\CE}$ be a family of CRMs. 
Furthermore, suppose that there exists an $(\CG,\CT)\in\CE$ such that $(\WR,\CT,\PW)$ is a conditionally atomless probability space given $\CG$, $(\WR,\CG,\PW)$ is atomless and $\rho_{\CT}:=\rho_{\{\emptyset,\WR\},\CT}$ is law-invariant.
Then the subfamily $(\rho_{\CH,\CT})_{\CH\in\CE(\CT)}$ is strongly consistent if and only if for each $\CH\in\CE(\CT)$ the CRM $\rho_{\CH,\CT}$ is of the form
\begin{equation}\label{eq:321}
\rho_{\CH,\CT}(X)=g_{\CH,\CT}\left(f_{u_{\CT}}^{-1}\big(\BEW{\PW}{u_{\CT}(X)}{\CH}\big)\right),\quad\text{for all } X\in\Linf{\CT}{d},
\end{equation}
where $u_{\CT}:\R^d\to\R$ is strictly increasing and continuous, $f_{u_{\CT}}^{-1}:\Imag f_{u_{\CT}}\to\R$ is the unique inverse function of $f_{u_{\CT}}:\R\to\R; x\mapsto u_{\CT}(x\vecone)$ and $g_{\CH,\CT}:\Linf{\CH}{}\to\Linf{\CH}{}$ is strictly antitone, $\CH$-local, fulfills the Lebesgue property and $0\in\Imag g_{\CH,\CT}$.\\
In particular, for any CRM of type \eqref{eq:321} we have that $g_{\CH,\CT}=f_{\rho_{\CH,\CT}}$, where $f_{\rho_{\CH,\CT}}$ is defined in \thref{def:f}.
\end{Theorem}


Note that the latter results, being extensions from the two-CRM-case of the previous sections,  only used the strong consistency as a pairwise strong consistency of the elements in  subfamilies $(\rho_{\CH,\CT})_{(\CH,\CT)\in \CE(\CT)}$ of $(\rho_{\CH,\CT})_{(\CH,\CT)\in \CE}$. But if $\CI_2$ contains more than just one $\sigma$-algebra, then the definition of strong consistency given in \th\ref{d:consforfamily} also has implications on the relations between these subfamilies corresponding to different sets $\CE(\CT)$ for $\CT\in \CI_2$.

\begin{Assumption}\th\label{ass.2}
In order to have sufficiently many subfamilies we suppose for the remainder of this section that $\CI_1=\CI_2=:\CI.$
\end{Assumption}

\begin{Proposition}\th\label{p:intercons}
Let $(\rho_{\CH,\CT})_{(\CH,\CT)\in\CE}$ be a strongly consistent family such that \eqref{eq:321} holds for all $(\CH,\CT)\in\CE$. 
Then for all $\CT_1,\CT_2\in\CI$ and $\CH\in\CT_1\cap\CT_2$, $\CH\in \CI$,
$$\rho_{\CH,\CT_1}(X)=f_{\rho_{\CH,\CT_2}}\left(f_{u_{\CT_2}}^{-1}\big(a_{\CT_1,\CT_2}\BEW{\PW}{u_{\CT_1}(X)}{\CH}+b_{\CH,\CT_1,\CT_2}\big)\right),$$
where $a_{\CT_1,\CT_2}\in\R^+\backslash\{0\}$, $b_{\CH,\CT_1,\CT_2}\in\Linf{\CH}{}$ and $\BEW{\PW}{b_{\CH,\CT_1,\CT_2}}{\CG}=b_{\CG,\CT_1,\CT_2}$ for all $\CG\in\CI$ with $\CG\subseteq\CH$.
\end{Proposition}

In order to prove \thref{p:intercons} we need some auxiliary lemmas and therefore the proof is deferred to the end of this section.
From \thref{p:intercons} it follows that any strongly consistent family $(\rho_{\CH,\CT})_{(\CH,\CT)\in\CE}$ (under \thref{ass.2}) is basically a family of conditional certainty equivalents as in \cite{Frittelli2011b}:

\begin{Corollary}
In the situation of \thref{p:intercons}, if $a_{\CT_1,\CT_2}=1$, $b_{\CH,\CT_1,\CT_2}=0$ for all $\CH\subseteq\CT_1\cap\CT_2$ where $\CH\in \CI$ and $\CT_1,\CT_2\in\CI$, and if $\rho_{\CT,\CT}$ are normalized on constants for all $\CT\in\CI$, then $(\rho_{\CH,\CT})_{(\CH,\CT)\in\CE}$ satisfies
\begin{equation}\label{eq:multiCCE}
\rho_{\CH,\CT}(X)=-f^{-1}_{u_\CH}\big(\BEW{\PW}{u_{\CT}(X)}{\CH}\big),\quad X\in\Linf{\CT}{d}.
\end{equation}
\end{Corollary}
\begin{proof}
	If $a_{\CT_1,\CT_2}=1$ and $b_{\CH,\CT_1,\CT_2}=0$ for all $\CH\subseteq\CT_1\cap\CT_2$, then 
	$$\rho_{\CH,\CT_1}(X)=f_{\rho_{\CH,\CT_2}}\left(f_{u_{\CT_2}}^{-1}\big(\BEW{\PW}{u_{\CT_1}(X)}{\CH}\big)\right),$$
	and thus by choosing $\CT_2=\CH$ and since $\rho_{\CH,\CH}$ is normalized on constants we get \eqref{eq:multiCCE}.
\end{proof}
Next we prepare the proof of \thref{p:intercons}:
\begin{Lemma}\th\label{l:54}
Let $u:\R^d\to\R$ be a deterministic utility, i.e.\ $u$ is strictly increasing and continuous, and let $\CG$ and $\CH$ be a sub-$\sigma$-algebras of $\CF$ such that $\CG \subseteq\CH$. 
Then $$\BEW{\PW}{u(\Linf{\CH}{d})}{\CG}=u(\Linf{\CG}{d}).$$
\end{Lemma}
\begin{proof}
"$\supseteq$": Obvious. "$\subseteq$": Define the CRM $\rho_\CG:\Linf{\CH}{d}\to\Linf{\CG}{};X\mapsto-\BEW{\PW}{u(X)}{\CG}$. By \thref{l:0} it follows that 
\begin{align*}
\BEW{\PW}{u(\Linf{\CH}{d})}{\CG}&=-\rho_\CG(\Linf{\CH}{d})=-f_{\rho_\CG}(\Linf{\CG}{})=\BEW{\PW}{u(\Linf{\CG}{}\vecone)}{\CG}\\&\subseteq \BEW{\PW}{u(\Linf{\CG}{d})}{\CG}=u(\Linf{\CG}{d}).
\end{align*}
\end{proof}

\begin{Lemma}\th\label{l:affine2}
For an arbitrary $\CT\in\CI$ let $u_\CT:\R^d\to\R$ be a deterministic utility and define $\CX_\CH:=u_\CT(\Linf{\CH}{d})$ for all $\CH\in\CE(\CT)$. Moreover, let $p_\CH:\CX_\CH\to\Linf{\CH}{}$ be functions such that $p_\CH$ is $\CH$-local, strictly isotone and fulfills the Lebesgue-property. If for all $\CG,\CH\in\CE(\CT)$ with $\CG\subseteq\CH$ and $\CH$ atomless it holds that
\begin{equation}\label{eq:tower} p_\CG\left(\BEW{\PW}{F}{\CG}\right)=\BEW{\PW}{p_\CH(F)}{\CG}\text{ for all }F\in\CX_\CH,\end{equation}
then 
$$p_\CH(F)=a F+\beta_\CH,$$
where $a\in\R^+\backslash\{0\}$ and $\beta_\CH\in\Linf{\CH}{}$ such that $\BEW{\PW}{\beta_\CH}{\CG}=\beta_\CG$.\\
Note that \eqref{eq:tower} is well-defined by \th\ref{l:54}.
\end{Lemma}

\begin{proof}
Firstly, we consider the case where $\CG$ is the trivial $\sigma$-algebra. We write $p:=p_{\{\WR,\emptyset\} }$.
Note that, since $p$ is a deterministic function, $p\left(\EW{\PW}{F}\right)$ is law-invariant and thus by \eqref{eq:tower} also $\EW{\PW}{p_\CH(F)}$.\\
Now suppose that there exist $x,y\in\CX:=\CX_{\{\WR,\emptyset\}}$ with $p_\CH(x)-p_\CH(y)\not\in\R$, i.e.\ there exists a $c\in\R$ such that 
$\PW(p_\CH(x)\leq p_\CH(y)+c)\in(0,1)$.
Since $\CH$ is an atomless space we can choose $A_1,A_2,A_3\in\CH$ with
$$\PW(A_1)=\PW(A_2):=q>0$$
such that 
$$A_1\subseteq\{p_\CH(x)\leq p_\CH(y)+c\},A_2\subseteq\{p_\CH(x)> p_\CH(y)+c\},A_3:=(A_1\cup A_2)^C.$$ 
Moreover, we define
$$F_1:=x\ind_{A_1}+y\ind_{A_2}+x\ind_{A_3}\quad\text{and}\quad F_2:=y\ind_{A_1}+x\ind_{A_2}+x\ind_{A_3}.$$
Obviously $F_1,F_2\sim q\delta_y+(1-q)\delta_x$, that is $F_1\eqd F_2$. However, since $p_\CH$ is $\CH$-local, we have
\begin{align*}
\EW{\PW}{p_\CH(F_1)}+cq&=\EW{\PW}{p_\CH(x)\ind_{A_1}}+\EW{\PW}{(p_\CH(y)+c)\ind_{A_2}}+\EW{\PW}{p_\CH(x)\ind_{A_3}}\\
&<\EW{\PW}{(p_\CH(y)+c)\ind_{A_1}}+\EW{\PW}{p_\CH(x)\ind_{A_2}}+\EW{\PW}{p_\CH(x)\ind_{A_3}}\\
&=\EW{\PW}{p_\CH(F_2)}+cq,
\end{align*}
which contradicts the law-invariance of $F\mapsto\EW{\PW}{p_\CH(F)}$.\\
Hence we have that $p_\CH(x)-p_\CH(y)\in\R$ for all $x,y\in\CX$. 
Choose an arbitrary $\widetilde x\in\CX$, and let $$a(x):=p_\CH(x)-p_\CH(\widetilde x),\quad x\in\CX, $$
so $a:\CX\to\R$. Define $\widetilde\beta_\CH:=p_\CH(\widetilde x)\in\Linf{\CH}{}$, then 
$p_\CH(x)=a(x)+\widetilde\beta_\CH$.
The function $a$ is continuous, since otherwise there would exist a sequence $(x_n)_{n\in\N}\subset\CX$ with $x_n\to x\in\CX$, but $a(x_n)\not\to a(x)$ and the Lebesgue-property would imply the contradiction 
$$p_\CH(x)=\lim_{n\to\infty}p_\CH(x_n)=\lim_{n\to\infty}a(x_n)+\widetilde{\beta}_\CH\neq a(x)+\widetilde\beta_\CH=p_\CH(x).$$
Let $F\in\CX_\CH$. Since the $\CH$-measurable simple random vectors are dense in $\Linf{\CH}{d}$ and by the definition of $\CX_\CH$ there exists a sequence of $\CH$-measurable simple random variables $(F_n)_{n\in\N}\subset\CX_\CH\cap \CS$ with $F_n=\sum_{i=1}^{k_n}x_i^n1_{A_i^n}\to F$ $\PW$-a.s. Thus
\begin{align*}
p_\CH(F)&=\lim_{n\to\infty}p_\CH(F_n)=\lim_{n\to\infty}\sum_{i=1}^{k_n} p_\CH(x_i^n)\ind_{A_i^n}=\lim_{n\to\infty}\sum_{i=1}^{k_n} a(x_i^n)\ind_{A_i^n}+\widetilde\beta_\CH\\&=\lim_{n\to\infty} a\left(\sum_{i=1}^{k_n} x_i^n\ind_{A_i^n}\right)+\widetilde\beta_\CH=\lim_{n\to\infty}  a(F_n)+\widetilde\beta_\CH= a(F)+\widetilde\beta_\CH.
\end{align*}
The function $\CX_\CH\ni F\mapsto\EW{\PW}{F}$ induces a preference relation on $\CM:=\{\mu: \exists F\in \CX_\CH$ such that $F\sim \mu\}$ via 
$$\mu\prefto\nu\quad\Longleftrightarrow\quad\EW{\PW}{F}\geq\EW{\PW}{G},F\sim\mu,G\sim\nu.$$
Moreover the function $ x\mapsto p^{-1}(x+\mathbb{E}[\widetilde\beta_\CH])$ is strictly increasing and by \eqref{eq:tower}
$$\EW{\PW}{F}=p^{-1}\left(\EW{\PW}{p_\CH(F)}\right)=p^{-1}\left(\EW{\PW}{ a(F)}+\mathbb{E}\big[\widetilde\beta_\CH\big]\right).$$
Thus $\EW{\PW}{a(F)}$ is another affine numerical representation of $\prefto$. It is well-known that the affine numerical representation of $\prefto$ is unique up to a positive affine transformation (see e.g.\ \cite{Follmer2011}~Theorem~2.21), i.e.\ there exist $\tilde a,b\in\R, \tilde a>0$ such that $\EW{\PW}{a(F)}=\tilde a \EW{\PW}{F}+b$ for all $F\in\CX_\CH$.
In particular this implies that for all $x\in\CX$
$$a(x)=\EW{\PW}{a(x)}=a\EW{\PW}{x}+b=\tilde ax+b.$$ By setting $b+\widetilde\beta_\CH=:\beta_\CH\in\Linf{\CH}{}$ we get for all $F\in\CX_\CH$ that
$$p_\CH(F)= a(F)+\widetilde\beta_\CH=\tilde a F+b+\widetilde\beta_\CH=\tilde a F+\beta_\CH.$$
Finally we obtain by \eqref{eq:tower} that for every $\CG\subseteq\CH$ and for all $F\in\CX_\CG$
$$p_\CG(F)=p_\CG\left(\BEW{\PW}{F}{\CG}\right)=\BEW{\PW}{p_\CH(F)}{\CG}=a F+\BEW{\PW}{\beta_\CH}{\CG},$$
which proves the martingale property of $(\beta_\CG)_{\CG\subseteq\CH}$.
\end{proof}

\begin{proof}[Proof of \thref{p:intercons}:]
Let $(\rho_{\CH,\CT})_{(\CH,\CT)\in\CE}$ be a strongly consistent family such that \eqref{eq:321} holds for all $(\CH,\CT)\in\CE$, i.e.
$$\rho_{\CH,\CT}(X)=f_{\rho_{\CH,\CT}}\left(f_{u_{\CT}}^{-1}\big(\BEW{\PW}{u_{\CT}(X)}{\CH}\big)\right),\quad\text{for all } X\in\Linf{\CT}{d},$$
We define the functions 
$$h_{\CH,\CT}:u_{\CT}(\Linf{\CH}{d})\to\Linf{\CH}{};F\mapsto f_{\rho_{\CH,\CT}}\circ f_{u_{\CT}}^{-1}(F)$$
and 
$$p_{\CH,\CT_1,\CT_2}: u_{\CT_1}(\Linf{\CH}{d})\to\Linf{\CH}{}; F\mapsto h^{-1}_{\CH,\CT_2}\circ h_{\CH,\CT_1}(F).$$
By strong consistency, we obtain for $\CG\subseteq\CH\subseteq\CT_1\cap\CT_2$, $X\in\Linf{\CT_1}{d}$ and $F:=\BEW{\PW}{u_{\CT_1}(X)}{\CH}$ that
\begin{align}
p_{\CG,\CT_1,\CT_2}\left(\BEW{\PW}{F}{\CG}\right)&=h^{-1}_{\CG,\CT_2}\left(h_{\CG,\CT_1}\big(\BEW{\PW}{\BEW{\PW}{u_{\CT_1}(X)}{\CH}}{\CG}\big)\right)\notag\\
&=h^{-1}_{\CG,\CT_2}\left(\rho_{\CG,\CT_1}(X)\right)\notag\\
&=h^{-1}_{\CG,\CT_2}\left(\rho_{\CG,\CT_2}\left(f^{-1}_{\rho_{\CH,\CT_2}}\big(\rho_{\CH,\CT_1}(X)\big)\vecone\right)\right)\notag\\
&=\BEW{\PW}{h^{-1}_{\CH,\CT_2}\left(h_{\CH,\CT_1}\big(\BEW{\PW}{u_{\CT_1}(X)}{\CH}\big)\right)}{\CG}\notag\\
&=\BEW{\PW}{p_{\CH,\CT_1,\CT_2}(F)}{\CG}.\label{eq:532}
\end{align}
By \th\ref{l:affine2} \eqref{eq:532} is fulfilled, if and only if 
$$p_{\CH,\CT_1,\CT_2}(F)=a_{\CT_1,\CT_2}F+b_{\CH,\CT_1,\CT_2},\quad \text{for all } F\in u_{\CT_1}(\Linf{\CH}{d}),$$
where $a_{\CT_1,\CT_2}\in\R^+\backslash\{0\}$, $b_{\CH,\CT_1,\CT_2}\in\Linf{\CH}{}$ and $\BEW{\PW}{b_{\CH,\CT_1,\CT_2}}{\CG}=b_{\CG,\CT_1,\CT_2}$ for all $\CG\in\CI$ with $\CG\subseteq\CH$. Thus
\begin{equation*}
h_{\CH,\CT_1}(F)=h_{\CH,\CT_2}(a_{\CT_1,\CT_2}F+b_{\CH,\CT_1,\CT_2}),\quad F\in u_{\CT_1}(\Linf{\CH}{d}),
\end{equation*}
which implies that
$$\rho_{\CH,\CT_1}(X)=f_{\rho_{\CH,\CT_2}}\left(f_{u_{\CT_2}}^{-1}\big(a_{\CT_1,\CT_2}\BEW{\PW}{u_{\CT_1}(X)}{\CH}+b_{\CH,\CT_1,\CT_2}\big)\right).$$
\end{proof}

\bibliographystyle{chicago}

\end{document}